\documentclass[11pt]{article}

\usepackage{fullpage}
\usepackage{times}
\usepackage{soul}
\usepackage{url}
\usepackage[utf8]{inputenc}
\usepackage[small]{caption}
\usepackage{graphicx}
\usepackage{amsmath}
\usepackage{booktabs}

\usepackage[ruled,linesnumbered]{algorithm2e}
\usepackage{ifthen}
\SetArgSty{textrm}

\usepackage{enumitem}

\usepackage{libertine}

\usepackage{amssymb,amsmath,amsfonts,amstext,amsthm}

\usepackage[breaklinks=true]{hyperref}
\usepackage[svgnames]{xcolor}
\usepackage[capitalise,nameinlink]{cleveref}
\hypersetup{colorlinks={true},linkcolor={DarkBlue},citecolor=[named]{DarkGreen}}

\usepackage{tikz}  
\usetikzlibrary{arrows}
\usetikzlibrary{patterns,snakes}
\usetikzlibrary{decorations.shapes}
\tikzstyle{overbrace text style}=[font=\tiny, above, pos=.5, yshift=5pt]
\tikzstyle{overbrace style}=[decorate,decoration={brace,raise=5pt,amplitude=3pt}]
\usetikzlibrary{shapes.geometric}
\usetikzlibrary{shapes,positioning}
\usetikzlibrary{calc,positioning}
\usetikzlibrary{shapes.multipart}

\definecolor{cadmiumgreen}{rgb}{0.0, 0.42, 0.24}

\usepackage{bbm}

\newtheorem{theorem}{Theorem}

\newtheorem{lemma}{Lemma}

\newcommand{\BAG}{\mathsf{BAG}}
\newcommand{\opt}{\mathsf{OPT}}
\newcommand{\sw}{\mathsf{SW}}
\newcommand{\PoF}{\mathsf{PoF}}
\newcommand{\GMMS}{\mathsf{GMMS}}

\theoremstyle{definition}
\newtheorem{open}{Open Problem}
\newtheorem{defn}{Definition}
\newtheorem{ex}{Example}

\usepackage{natbib}

\usepackage{mathtools}

\newcommand{\newtext}[1]{\textcolor{black}{#1}}

\usepackage{authblk}

\usepackage{subcaption}

\usepackage{mdframed}
\mdfsetup{linewidth=1pt}

\usepackage{xspace}
\usepackage{fancyhdr}
\usepackage{tcolorbox}

\begin{document}

\allowdisplaybreaks

\title{\bf Fair Division of Indivisible Goods: \\ Recent Progress and Open Questions\thanks{Preliminary versions of this survey appeared as \citep{fair-survey-ABFV22} and \citep{fair-survey-ALMW22}.}}

\author[1]{Georgios Amanatidis}
\author[2]{Haris Aziz}
\author[3]{Georgios Birmpas}
\author[4]{Aris Filos-Ratsikas}
\author[5]{Bo Li}
\author[6]{Herv\'e Moulin}
\author[1]{Alexandros A. Voudouris}
\author[7]{Xiaowei Wu}

\affil[1]{University of Essex}
\affil[2]{UNSW Sydney}
\affil[3]{Sapienza University of Rome}
\affil[4]{University of Edinburgh}
\affil[5]{The Hong Kong Polytechnic University}
\affil[6]{University of Glasgow } 
\affil[7]{University of Macau}

\date{}

\maketitle

\begin{abstract}
Allocating resources to individuals in a \emph{fair} manner has been a topic of interest since ancient times, with most of the early mathematical work on the problem focusing on resources that are infinitely divisible. Over the last decade, there has been a surge of papers studying computational questions regarding the \emph{indivisible} case, for which exact fairness notions such as envy-freeness and proportionality are hard to satisfy. One main theme in the recent research agenda is to investigate the extent to which their relaxations, like maximin share fairness (MMS) and envy-freeness up to any good (EFX), can be achieved. In this survey, we present a comprehensive review of the \newtext{recent} progress made in the related literature by highlighting different ways to relax fairness notions, common algorithm design techniques, and the most interesting questions for future research.
\end{abstract}

\section{Introduction}

Fair division is concerned with the fundamental task of \emph{fairly} partitioning or allocating a set of resources to people with diverse and heterogeneous preferences over these resources.
Some examples of fair division in the real world include Course Match, which is employed for course allocation at the Wharton School in the University of Pennsylvania, and the websites Spliddit (\url{spliddit.org}) and Fair-Outcomes (\url{fairoutcomes.com}), that provide implementations of fair allocation algorithms for various allocation problems, such as sharing rent among roommates, splitting taxi fares, and for assigning goods to a set of individuals (which is the focus of this survey).
The associated theory originated in the works of \citet{Steinhaus49}, Banach, and Knaster (see \citep{dubins1961cut}), and has been in the focus of economics, ma\-the\-ma\-tics and computer science for the better part of the last century \citep{books/daglib/0017734}. Most of the classic work on the problem has been devoted to the fair division of \emph{infinitely divisible} resources, where  ``fair'' here may have different interpretations, with two predominant ones being \emph{proportionality} \mbox{\citep{Steinhaus49}} and \emph{envy-freeness} \citep{GS58,Varian74}; \newtext{see also \citep{Procaccia_cake_16,LindnerR16}.}
In a recent breakthrough, \citet{conf/focs/AzizM16} showed that an envy-free allocation of divisible resources (which is also proportional in the standard additive case) can always be found in a finite number of steps. 

Compared to the divisible setting, the fair division of \emph{indivisible} resources, referred to as \emph{discrete fair division}, turns out to be inherently more challenging. Indeed, no reasonable fair solution can be guaranteed in some cases; for example, when there is a single valuable item, no matter who gets it, the allocation is not fair to the others. A typical remedy to this situation is to employ ran\-do\-mi\-zation, and aim for fairness (such as envy-freeness) in expectation \citep{hylland1979efficient,journals/jet/BogomolnaiaM01}. 
 A fundamentally different approach to discrete fair division came via the introduction of appropriate \emph{relaxations} of envy-freeness and proportionality, originating in the works of \citet{LMMS04}, \citet{Budish11}, \citet{CaragiannisKMPS19}, and \citet{GourvesMT14}, which are geared to escape such adverse examples. The main notions that were introduced in this literature are \emph{envy-freeness up to one good} (EF1), \emph{envy-freeness up to any good} (EFX) and \emph{maximin share fairness} (MMS). Since then, work on the topic has flourished, centered around fundamental questions about the existence and the efficient computation of allocations satisfying these or other related fairness criteria. 

\newtext{We remark that the problem of fairly dividing indivisible goods has a very long history, and our aim in this survey is not to provide an extensive review of all of this literature. Instead, our focus is on the recent developments over the past decade, with an emphasis on the algorithmic aspects of the associated problems. For other results,} there are several surveys highlighting different perspectives of the theory of fair resource allocation. 
\citet{moulin2018fair} provides a review from an economics perspective,  
\citet{aleksandrov2020onlinesurvey} and \citet{suksompong2021constraints} focus on online and constrained settings, respectively, \newtext{whereas \citet{conf/ijcai/Walsh20} and
\citet{conf/aaai/Aziz20} provide short surveys  targeted more towards a broader audience.
\citet{LangR16}, \citet{BouveretCM16}, and \citet{Markakis17-survey} consider general discrete fair division settings, and thus are more closely related to this work. 
Our survey differs from these in that we pay particular attention on algorithmic progress made in the last decade, focusing on scenarios where it is assumed that the agents are equipped with valuation functions (most often additive ones) on the sets of items.}
Over this period of time, discrete fair division has been at the epicenter of computational fair division, for several different fairness notions and a variety of different settings. In this survey, we highlight the main contributions of this literature, the most significant variants of the main setting, common algorithmic techniques, as well as some of the major open problems in the area. \newtext{ Although they lack most of this content, the surveys of \citet{LangR16}, and \citet{BouveretCM16} remain  excellent starting points on questions of preference elicitation and compact representation.}

\paragraph{Roadmap.}
The rest of the survey is organized as follows. 
We first introduce the problem of discrete fair division and the main fairness notions in Section~\ref{sec:model}.
We then discuss results on EF1 in Section~\ref{sec:EF1}, EFX in Section~\ref{sec:EFX}, and MMS in Section~\ref{sec:MMS}.
In Section~\ref{sec:other-notions} we introduce other notable fairness notions for the main setting. In Section \ref{sec:efficiency&truthfulness}, we talk about two more properties, efficiency and truthfulness, that are usually desirable together with fairness.
Finally, Section \ref{sec:other_settings} is dedicated to other settings, including limited information, general valuations, randomness, and more.

\section{The Setting}
\label{sec:model}
For the general discrete fair division problem we consider here, there is a set $N$ of $n$ {\em agents} and a set $M$ of $m$ {\em goods} that cannot be divided or shared. Each agent $i\in N$ is equipped with a {\em valuation function} $v_i:2^M \rightarrow \mathbb{R}_{\geq 0}$, which assigns a non-negative real number to each possible subset of items and is \emph{normalized} and \emph{monotone}, i.e., $v_i(\varnothing) = 0$ and $v_i(S) \le v_i(T)$ for all $S \subseteq T \subseteq M$. In this survey, we mostly focus on the case where the valuation function of each agent $i$ is also assumed to be {\em additive}, so that $v_i(S) = \sum_{g \in S} v_i(g)$ for any subset of items $S \subseteq M$, where $v_i(g)$ is used as a shorthand for $v_i(\{g\})$. Other types of valuation functions have also been studied and are briefly discussed but, unless otherwise specified, in what follows, we refer to the additive case. \newtext{For the sake of readability, we avoid defining all different function classes mentioned here; we define submodular, XOS, and subadditive valuation functions in Section \ref{ssec:general-valuations}, and we refer the interested reader to the corresponding references for the others.}
A fair allocation instance is denoted by $I=(N,M,v)$ where $v=(v_1,\ldots,v_n)$ is the vector of valuation functions and can be represented by a table with a row per agent and a column per good, such that cell $(i,g)$ contains the value $v_i(g)$. 
An \emph{allocation} is a tuple of subsets of $M$, $A=(A_1,\ldots,A_n)$, such that each agent $i \in N$ receives the \emph{bundle}, $A_i \subseteq M$, $A_i \cap A_j = \varnothing$ for every pair of agents $i,j \in N$, and $\bigcup_{i \in N}A_i = M$. If $\bigcup_{i \in N}A_i \subsetneq M$, the allocation is called {\em partial}.

\subsection{Solution Concepts}
The objective is to compute a {\em fair} allocation, i.e., an allocation that satisfies a desired fairness criterion. 
As already mentioned, since the early fair division literature, there are two predominant fairness notions, namely {\em envy-freeness} and {\em proportionality}. An allocation is said to be envy-free if no agent believes that another agent was given a better bundle; note that envy-freeness depends on pairwise comparisons. 

\begin{defn}[Envy-freeness]
\label{def:ef}
An allocation $A$ is {\em envy-free} (EF) if $v_i(A_i) \geq v_i(A_j)$ for every pair of agents $i, j \in N$.
\end{defn}

\noindent
On the other hand, an allocation is said to be proportional if each agent is guaranteed her \emph{proportional share} in terms of the total value, independently of what others get. 

\begin{defn}[Proportionality]
\label{def:prop}
An allocation $A$ is {\em proportional} (PROP) if $v_i(A_i) \geq 
{v_i(M)}/{n}$ 
for every agent $i \in N$.
\end{defn}

\noindent 
It is not hard to see that an EF allocation is also PROP, but the converse is not necessarily true.

\begin{ex} \label{ex}
Consider an instance with three agents and four goods. The values of the agents for the goods are given in Table \ref{table:def:instance}. 

\begin{table}[ht]
    \centering
    \begin{tabular}{c|cccc}
                & $g_1$ & $g_2$ & $g_3$ & $g_4$\\
        \hline
        $a_1$   & 10    & 6     & 6     & 8  \\
        $a_2$   & 10    & 5     & 5     & 10 \\
        $a_3$   & 10    & 0     & 10    & 10 \\
        \hline
    \end{tabular}
    \caption{The values of the agents in the instance considered in Example~\ref{ex}. In particular, the number in the cell defined by row $i \in \{a_1,a_2,a_3\}$ and column $g \in \{g_1, g_2, g_3, g_4\}$ is the value $v_i(g)$ of agent $i$ for good $g$.}
    \label{table:def:instance}
\end{table}

The allocation $A=\{A_1,A_2,A_3\}$ with $A_1 = \{g_2,g_3\}$, $A_2 = \{g_4\}$ and $A_3 = \{g_1\}$ is EF. In particular, agent $a_1$ does not envy any other agent as $v_1(A_1) = 12 > v_1(A_2)=8$ and $v_1(A_1) = 12 > v_1(A_3)=10$, agent $a_2$ does not envy as $v_2(A_2) = 10 = v_2(A_1) = v_2(A_3)$, and agent $a_3$ does not envy as $v_3(A_3) = 10 = v_3(A_1) = v_3(A_2)$. 

On the other hand, the allocation $B=\{B_1,B_2,B_3\}$ with $B_1=\{g_1\}$, $B_2 = \{g_2,g_3\}$ and $B_3 = \{g_4\}$ is not EF, but it is PROP. Indeed, agent $a_1$ envies agent $2$: $v_1(B_1) = 10 < v_1(B_2) =12$. However, $v_i(B_i) \ge v_i(M)/3 = 10$ for every agent $a_i$, and thus each agent obtains her proportional share.
\hfill $\qed$
\end{ex}

As mentioned above, EF and PROP allocations do not always exist when allocating indivisible items. The simplest example is the case of two agents and a single good that is positively valued by both agents. Since only one of the agents receives the good, the other agent gets zero value, she envies the agent with the item, and does not achieve her proportional share. See also Example~\ref{ex:ef1} for another instance that does not admit any EF or PROP allocation.
Despite this impossibility, one could still be interested in finding EF or PROP allocations \emph{when} they exist. Unfortunately, it turns out that the problem of even deciding whether an instance admits an EF or PROP allocation is NP-complete, which can be shown via a simple reduction from \textsc{Partition} \citep{LMMS04}. 
These straightforward impossibility results have led to the definition of multiple relaxations of these two notions, tailored for discrete fair division. 

\subsection{Important Relaxations}
The first relaxation of envy-freeness is {\em envy-freeness up to one good} (EF1), implicitly introduced by \citet{LMMS04}, but formally defined by \citet{Budish11}. According to EF1 it is acceptable for an agent $i$ to envy another agent $j$, as long as there exists a good in $j$'s bundle the hypothetical removal of which would eliminate $i$'s envy towards $j$.

\begin{defn}[EF1] \label{def:ef1}
An allocation $A$ is {\em envy-free up to one good (EF1)} if, for every pair of agents $i,j \in N$, it holds that $v_i(A_i) \geq v_i(A_j \setminus \{g\})$ for some $g \in A_j$.
\end{defn}

\begin{ex} \label{ex:ef1}
To demonstrate the notion of EF1 (as well as EFX and MMS later on), let us consider a simple example with three agents and five goods. The values of the agents for the goods are given in Table \ref{table:example:EF1}.

\begin{table}[ht]
    \centering
\begin{tabular}{c|ccccc}
        & $g_1$ & $g_2$ & $g_3$ & $g_4$ & $g_5$ \\\hline
$a_1$   &   15  & 3     & 2     &  2    &  6\\
$a_2$   &   7   & 5     & 5     &  5    &  7\\
$a_3$   &   20  & 3     & 3     &  3    &  3\\\hline
\end{tabular}
\caption{The instance considered in Examples~\ref{ex:ef1}, \ref{ex:efx}, and \ref{ex:mms}.}
\label{table:example:EF1}
\end{table}
This instance does not admit any EF or PROP allocations.  To see this, observe that in any PROP allocation, agent $a_3$ must get at least $\{g_1\}$ or $\{g_2, g_3, g_4, g_5\}$. In the latter case, at least one of $a_1$ and $a_2$ will get no goods, whereas in the former case $a_1$ must get at least three of the remaining four goods and $a_2$ must get at least two, which is not possible.
On the other hand, note that the allocation $A_1 = \{g_3, g_4\}$, $A_2 = \{g_2, g_5\}$,  $A_3 = \{g_1\}$ is EF1: $a_2$ and $a_3$ are not envious, and the envy of $a_1$ towards $a_2$ and $a_3$ can be eliminated by the hypothetical removals of $g_5$ from $A_2$ and $g_1$ from $A_3$, respectively. 
\hfill $\qed$
\end{ex}

As we will see, EF1 is easy to achieve, even when the valuation functions are general monotone. 
However, in many cases it is a fairly weak fairness notion; an EF1 allocation is considered to be fair for an agent even when a very highly-valued good is hypothetically removed from another agent's bundle (e.g., a house or an expensive car). For example, consider agent $a_1$'s perspective of the allocation in Example \ref{ex:ef1}, where the proposed EF1 solution requires the removal of rather valuable goods for the agent. A very natural refinement of the notion is the stricter relaxation of {\em envy-freeness up to any good (EFX)} that was introduced in 2016 by \citet{CaragiannisKMPS19} in the conference version of their work but also somewhat earlier by \citet{GourvesMT14} under the name {\em near envy-freeness}. An allocation is said to be EFX if the envy of an agent $i$ towards another agent $j$ can be eliminated by the hypothetical removal of {\em any} good in $j$'s bundle. 

\begin{defn}[EFX] \label{def:efx}
An allocation $A$ is {\em envy-free up to any good (EFX)} if, for every pair of agents $i,j \in N$, it holds that $v_i(A_i) \geq v_i(A_j \setminus \{g\}$) for any $g\in A_j$.
\end{defn}

\begin{ex} \label{ex:efx}
Consider the instance of Example~\ref{ex:ef1} again. The allocation $A_1 = \{g_3, g_4\}$, $A_2 = \{g_2, g_5\}$,  $A_3 = \{g_1\}$ is not EFX since the envy of $a_1$ towards $a_2$ cannot be eliminated by removing $g_2$ ($a_1$'s least favorite good in $A_2$) from $A_2$. Nevertheless, it is easy to modify this allocation to get $B_1 = \{g_4, g_5\}$, $B_2 = \{g_2, g_3\}$, $B_3 = \{g_1\}$ that is EFX. Indeed, the envy of $a_1$ towards $a_3$ can be eliminated by removing $g_1$ from $B_3$, whereas the envy of $a_2$ towards $a_1$ can be eliminated by removing $g_4$ from $B_1$; in both cases, the hypothetical removal involves the envious agent's least valued good in the other agent's bundle.
\hfill $\qed$
\end{ex}

In contrast to EF1, the existence of EFX allocations is a challenging open problem.  In fact, \cite{procaccia2020technical} referred to this as ``fair division's most enigmatic question''.
In the past few years, a sequence of works have partially or approximately answered this question; see Section \ref{sec:EFX}. 

Besides the two additive relaxations of envy-freeness discussed so far, an extensively studied fairness notion in discrete fair division is {\em maximin share fairness}, also introduced by \cite{Budish11}. The notion can be seen as a generalization of the rationale of the well-known cut-and-choose protocol, which is known to guarantee an envy-free partition of a divisible resource among two agents. 
Here, the goal is to give each agent $i$ goods of value at least as much as her \emph{maximin share} $\mu^n_i(M)$, which is the maximum value this agent could guarantee for herself by partitioning the set of goods $M$ into $n$ disjoint bundles and keeping the worst of them. As such, it is a relaxation of proportionality. 

\begin{defn}[MMS]\label{def:mms}
Let $\mathcal{A}_n(M)$ be the collection of possible allocations of the goods in $M$ to $n$ agents. An allocation $A$ is said to be {\em maximin share fair (MMS)} if for each agent $i \in N$,
$$v_i(A_i) \geq \mu^n_i(M) = \!\!\displaystyle\max_{B \in \mathcal{A}_n(M)} \min_{S \in B} v_i(S).$$
\end{defn}

\noindent
When $M$ is clear from context, we may simple write $\mu_i^n$ for $\mu_i^n(M)$.

\begin{ex} \label{ex:mms}
Returning to the instance of Example~\ref{ex:ef1}, we can see that $\mu^3_1(M) = 6$, since it is not possible to partition the items into three sets with strictly more value, but $6$ is guaranteed by the partition $\{g_1\}$, $\{g_2, g_3, g_4\}$, $\{g_5\}$. Similarly, $\mu^3_2(M) = 7$ and $\mu^3_3(M) = 6$.
Therefore, $B_1 = \{g_4, g_5\}$, $B_2 = \{g_2, g_3\}$,  $B_3 = \{g_1\}$, from Example \ref{ex:efx}, is an MMS allocation, but $A_1 = \{g_3, g_4\}$, $A_2 = \{g_2, g_5\}$,  $A_3 = \{g_1\}$, from Example \ref{ex:ef1}, is not as agent $a_1$ gets a bundle of value only $4 = 2/3 \cdot \mu^3_1(M)$.
\hfill $\qed$
\end{ex}

While it is relatively easy to see that computing MMS allocations or even computing the maximin share of an agent is an NP-hard problem using a reduction from {\sc Partition}, there is a PTAS for the latter task \citep{Woeginger97}. Since the introduction of the notion, the existence of MMS allocations was a very intriguing open problem. This was eventually answered in the negative by \citet{KurokawaPW18,KPW16}, who proved that MMS allocations do not always exist
when there are more than two agents. 
Still, it is possible to compute \emph{approximate} MMS allocations; see Section \ref{sec:MMS}.

\section{Envy-Freeness up to One Good (EF1)}
\label{sec:EF1}

We start our discussion by presenting important results about EF1 allocations.

\subsection{Computing EF1 Allocations}
There are several simple, polynomial-time algorithms for computing EF1 allocations. The simplest one is Round-Robin (Algorithm \ref{alg:round-robin}), which allocates the goods to the agents in multiple rounds. In each round, according to some arbitrarily fixed given ordering, each agent chooses her most valuable good among the available ones (that is, from the set of goods that have not been chosen by the agent's turn). It is not hard to see that Round-Robin leads to an EF1 allocation~\citep{CaragiannisKMPS19}.

\begin{algorithm}[ht]
\caption{Round-Robin}
\label{alg:round-robin}
\textbf{Input: }{A fair allocation instance $I = (N,M,v)$ with $n$ agents and $m$ goods.}\\
\textbf{Output: }{Allocation $A=(A_1,\ldots,A_n)$.} \\
\For{each agent $i \in N$}{
    $A_i \gets \varnothing$\; 
}
\For{$\ell = 1, \ldots, m$}{
    Let $i \gets \ell \!\mod n$\; 
    Let $g^* \in \arg \max_{g \in M} v_i(g)$\; 
    $A_i \gets A_i \cup \{g^*\}$\; 
    $M \gets M\setminus\{g^*\}$\;
}  
\end{algorithm}

\begin{theorem}[\cite{CaragiannisKMPS19}]
Round-Robin computes an EF1 allocation.
\end{theorem}

\begin{proof}
Consider two agents $i$ and $j$, such that $i$ chooses before $j$ according to the given fixed ordering. As $i$ has the chance to pick a good before $j$ in every single round of the algorithm, $i$ cannot envy $j$. Of course, agent $j$ may envy agent $i$. Let $g$ be the first good chosen by $i$. From that point on, we can consider the execution of the algorithm on the remaining goods as a fresh run where now $j$ has the chance to pick a good before $i$ in every round. So, $j$ does not envy $i$'s bundle after the removal of good $g$ from it.
\end{proof}

In the basic definition of Round-Robin, the agents follow the same order in each round to select the goods. However, to compute an EF1 allocation, this is not necessary; as long as each agent selects her favorite good when it is her turn, the order of the agents in different rounds can be different. Essentially, Round-Robin is only a member of a larger class of algorithms that compute EF1 allocations by using \textit{recursively balanced} sequences in different rounds, where the difference between the number of turns of any two agents is at most $1$. 

Another algorithm for computing an EF1 allocation is the Envy-Cycle Elimination (Algorithm~\ref{alg:envy-cycle}), introduced by \citet{LMMS04}. In contrast to Round-Robin or any other sequential allocation algorithm, Envy-Cycle Elimination does not use a prefixed sequence for agents to select goods. Instead, it repeatedly chooses an agent that is in a disadvantage compared to other agents and gives an unallocated good. In the variant we present in Algorithm~\ref{alg:envy-cycle}, whenever an agent gets a new good, this is her favorite available one.
The algorithm maintains an {\em envy graph}, where the nodes correspond to agents and there is an edge from agent $i$ to agent $j$ if $i$ is envious of $j$'s bundle. At each step of the algorithm, an unassigned good is allocated to some agent who is not envied by any other agent, i.e., an agent that corresponds to a node with in-degree $0$ in the envy graph. 
If no such agent exists, the envy graph must contain a directed cycle, which can be eliminated by redistributing the current bundles among the agents that participate in the cycle. Formally, let $C=(i_1,\ldots, i_d)$ be a directed cycle in the envy graph such that $i_j$ envies $i_{j+1}$ for each $j \in [d-1]$, and $i_d$ envies $i_1$. The cycle can be resolved by exchanging the bundles of items along the cycle: each agent in the cycle gets the bundle of the agent she points to (Equation \eqref{eq:envy-cycle} in the description of Algorithm~\ref{alg:envy-cycle}). Repeating this procedure, eventually, leads to a modified envy graph with at least one agent who is not envied by any other agent. The algorithm terminates when all items are allocated.

\begin{algorithm}[ht]
\caption{Envy-Cycle Elimination}
\label{alg:envy-cycle}
\textbf{Input: }{A fair allocation instance $I = (N,M,v)$ with $n$ agents and $m$ goods.}\\
\textbf{Output: }{Allocation $A=(A_1,\ldots,A_n)$.} \\

\For{each agent $i \in N$}{
    $A_i \gets \varnothing$\; 
}

\For{$\ell = 1,\ldots, m$}{
    \While{there does not exist an unenvied agent}
    {
        Find an envy-cycle $C = (i_1,\ldots, i_d)$ and resolve the cycle as follows:
        \begin{align} \label{eq:envy-cycle}
        A^C_{i_j} = 
        \begin{cases}
            A_{i_{j+1}} & \text{ for all $1 \le j \leq d-1$}\\
            A_{i_1} & \text{ for $j = d$}\;
        \end{cases} 
        \end{align}
        $A_{i} \gets A^C_{i}$ for all $i \in C$\;
    }
    
    Let $i$ be an unenvied agent\;
    Let $g^* \in \arg \max_{g \in M} \{v_i(g)\}$\label{line:most-valuable-item}\;
    $A_i \gets A_i \cup \{g^*\}$\;
    $M \gets M\setminus\{g^*\}$\;
    }
\end{algorithm}

\citet{LMMS04} proved that Envy-Cycle Elimination runs in polynomial time and outputs an EF1 allocation. Indeed, in each round of the algorithm, an agent receives a good only if she is not envied by any other agent. Therefore, by removing the last good an agent receives eliminates any possible envy of other agents towards her. Note that to ensure EF1, we can assign any good to agent $i$ (instead of the item $g^*$ with maximum value) in line~\ref{line:most-valuable-item}. However, as we will see later, by assigning the most valuable good, we can ensure other nice properties regarding EFX and MMS.

\subsection{EF1 and Pareto Optimal (PO) Allocations} \label{ssec:EF1-and-PO}
Besides computing a fair allocation, another natural requirement is that of efficiency. 
\citet{CaragiannisKMPS19} identified an interesting inherent connection between EF1, Pareto optimality, and the notion of {\em maximum Nash welfare (MNW)}.

\begin{defn} \label{def:po}
An allocation $A$ is said to be Pareto optimal (PO) if there is no allocation $B$ such that $v_i(B_i) \ge v_i(A_i)$ for all $i\in N$ and $v_j(B_j) > v_j(A_j)$ for some $j\in N$. Equivalently, we will say that such an allocation is not \emph{Pareto dominated} by any other allocation.
\end{defn}

\begin{defn}[MNW allocation] \label{def:mnw}
An allocation $A$ is said to be a {\em maximum Nash welfare (MNW) allocation} if 
(a) it maximizes the number of agents with positive value, and
(b) it maximizes the {\em Nash welfare}, defined as the product of values $\prod_i v_i(A_i)$ for the agents with positive value.\footnote{Defining the Nash welfare as the product of values is a simplification. Typically, the Nash welfare of an allocation $A$ is defined as the geometric mean of the values of the agents, that is, $\left( \prod_{i\in N} v_i(A_i) \right)^{1/n}$. Note that exactly maximizing the geometric mean of values is equivalent to exactly maximizing the product of values, but this is not true when discussing about approximation.}
\end{defn}

\begin{theorem}[\citet{CaragiannisKMPS19}] \label{thm:MNW:EF1+PO}
Every MNW allocation is EF1 and PO.
\end{theorem}

\begin{proof}
It is straightforward to show that every MNW allocation $A$ is PO. If there were a different allocation $B$ with $v_i(B)\geq v_i(A)$ for all $i\in N$ and $v_j(B_j) > v_j(A_j)$ for some $j\in N$ it would have strictly larger Nash welfare compared to that of $A$, which would contradict that $A$ is an MNW allocation.

To show that $A$ is EF1, suppose to the contrary that there exist agents $i$ and $j$ such that $v_i(A_i) < v_i(A_j \setminus \{ g \} )$ for every $g\in A_j$. 
We will show that there exists another allocation $A'$ with Nash welfare that is strictly larger than that of $A$. 
\newtext{Since agent $i$ envies agent $j$, there exists a $g^* = \arg\min_{g\in A_j, v_i(g)>0} \left\{ \frac{v_j(g)}{v_i(g)} \right\}$. Using this, we define the allocation $A'$ with $A'_\ell= A_\ell$ for each $\ell \in N\setminus \{i,j\}$, $A'_i = A_i \cup \{g^*\}$ and $A'_j = A_j\setminus \{g^*\}$.}
To show that the Nash welfare of $A'$ is strictly larger, it suffices to show that
\begin{align*}
    v_i(A'_i) \cdot v_j(A'_j) > v_i(A_i)\cdot v_j(A_j).
\end{align*}

By the definition of $g^*$ we have $\frac{v_j(g^*)}{v_i(g^*)} \leq \frac{v_j(A_j)}{v_i(A_j)}$, which implies
\begin{align} \label{eq:non-EF1}
    \frac{v_j(g^*)}{v_j(A_j)} \leq \frac{v_i(g^*)}{v_i(A_j)} < \frac{v_i(g^*)}{v_i(A_i) + v_i(g^*)},
\end{align}
where the last inequality follows from the fact that agent $i$ envies agent $j$ even after the removal of any item in $A_j$.
Therefore, we have
\begin{align*}
    v_i(A'_i) \cdot v_j(A'_j) & =
    \left( v_i(A_i) + v_i(g^*) \right)\cdot \left( v_j(A_j) - v_j(g^*) \right) \\
    & = \left( 1 + \frac{v_i(g^*)}{v_i(A_i)} \right)\cdot \left( 1 - \frac{v_j(g^*)}{v_j(A_j)} \right)\cdot v_i(A_i)\cdot v_j(A_j) \\
    & > \frac{v_i(A_i) + v_i(g^*)}{v_i(A_i)} \cdot \left( 1 - \frac{v_i(g^*)}{v_i(A_i) + v_i(g^*)} \right)\cdot v_i(A_i)\cdot v_j(A_j) \\
    & = v_i(A_i)\cdot v_j(A_j),
\end{align*}
where the first inequality follows from Inequality~\eqref{eq:non-EF1}.
Hence, the Nash welfare of $A'$ is strictly larger than that of $A$, a contradiction.
\end{proof}

This result shows that there exist allocations that combine fairness with (Pareto) efficiency. \newtext{\citet{yuen2023characterization} recently showed that Nash Welfare is the only welfarist function of agent valuations whose maximization leads to allocations that are EF1 and PO.} 

However, computing MNW allocations is not an easy task, as it is known that they are generally hard to even approximate in polynomial time; \newtext{see the relevant discussion in Section~\ref{sec:fair-PO}.} \citet{BarmanKV18} recently made progress by computing such allocations in pseudo-polynomial time, but it still remains unknown whether this can be done in polynomial time. 

We now mention our first open question. 

\begin{open}
Can an EF1 and PO allocation be computed in polynomial time?
\end{open}

\section{Envy-Freeness up to Any Good (EFX)}
\label{sec:EFX}

In contrast to EF1, where the existence is guaranteed via simple polynomial-time algorithms, the existence of EFX allocations is a challenging open problem. 
In the past few years, a sequence of works has positively answered this question for important special cases, centered mainly around two axes: a small number of agents or restricted agent valuations.
In this section, we first review the existence of EFX allocations for these special cases, and then discuss results about approximate versions of EFX.

\subsection{Existence and Computation of EFX Allocations for Special Cases}

We first discuss a few cases for which EFX allocations are known to exist.

\subsubsection{Identical Valuations}
The first positive existence result was shown by \citet{PR18} for when the agents have identical valuations, that is, $v_i(\cdot) = v(\cdot)$ for every agent $i$. In particular, \citeauthor{PR18} showed that allocations that satisfy particular properties about the minimum value are EFX.

A {\em leximin} allocation is an allocation that first maximizes the minimum value among all bundles, then maximizes the second minimum value, and so on. While a {\em leximin} allocation is not EFX, \citet{PR18} showed that a slight refinement of it is in fact EFX. The fix is that after maximizing the minimum value, the allocation must maximize the size of the bundle with minimum value, before maximizing the second minimum value, and so on. Such an allocation is called {\em leximin++}.

It is not hard to see that any leximin++ allocation $A$ must be EFX. Let $v(A_1) \le \ldots \le v(A_n)$ and suppose towards to the contrary that $A$ is not EFX; thus, there exist $i < j$ such that $v(X_i) < v(X_j \setminus\{g\})$ for any $g \in X_j$. Then, either the $i$-th minimum value can be strictly increased (if $v(g)>0$) or the size of this bundle can be increased (if $v(g)=0$), which contradicts that $A$ is leximin++. Note that a leximin++ allocation is EFX even when the common valuation function is not additive.

\subsubsection{Ordered Valuations}
Beyond identical valuations, \citet{PR18} also showed that the Envy-Cycle Elimination algorithm returns an EFX allocation for the slightly more general class of {\em ordered instances}, where the agents have identical orderings, but possibly distinct cardinal valuations. In particular, in such instances, there is an ordering of the goods $g_1,\dots,g_m$ such that $v_i(g_1) \ge \ldots \ge v_i(g_m)$ for every agent $i \in N$. The underlying intuition is that an agent can choose a good only when she is not envied by the others, and the newly added good is the least-valued (among the available ones) for every agent, which means that the removal of this good can eliminate any envy of the other agents towards the choosing agent.

\begin{theorem}[\cite{PR18}]
\label{thm:efx-for-identical}
The Envy-Cycle Elimination Algorithm (Algorithm~\ref{alg:envy-cycle}) computes an EFX allocation for every ordered instance.
\end{theorem}

\begin{proof}
Consider any agent $j$ and let $g$ be the last good agent $j$ is allocated. For any other agent $i$, we have $v_i(A_i) \geq v_i(A_j\setminus \{g\})$ because when good $g$ is allocated to agent $j$, $j$ is not envied by any other agent. Since the instance is ordered, for every good $g' \in A_j\setminus\{g\}$ (that is allocated before $g$), we have $v_i(g') \geq v_i(g)$.
Note that while some items in $A_j$ might be swapped in the envy cycle elimination procedure, we can still ensure $v_i(g') \geq v_i(g)$ because the instance is ordered and items with larger value are allocated first.
Hence, for every $g'\in A_j$, 
\begin{align*}
v_i(A_i) \geq v_i(A_j\setminus \{g\}) = v_i(A_j) - v_i(g) \geq  v_i(A_j) - v_i(g') = v_i(A_j\setminus \{g'\}).
\end{align*}
In other words, the allocation is EFX.
\end{proof}

\subsubsection{EFX for Two and Three Agents}
When there are only two agents, we can compute an EFX allocation using a \emph{divide-and-choose} approach that exploits the existence of EFX allocations for agents with identical valuations.

\begin{theorem}[\cite{PR18}]
EFX allocations always exist for instances with two agents.
\end{theorem}

\begin{proof}
We fix any agent $i$ and compute an EFX allocation (consisting of two bundles), pretending that both agents have the same valuation function $v_i$.
For additive functions we can use the Envy-Cycle Elimination algorithm, as we have shown in Theorem~\ref{thm:efx-for-identical}; for non-additive functions we can use the leximin++ allocation.
The other agent is allocated her favorite among the two bundles, and the remaining bundle is allocated to agent $i$. 
Clearly, since the other agent gets her favorite bundle, she does not envy. In addition, since the allocation used for the definition of the two bundles is EFX to agent $i$, no matter which bundle is left for agent $i$, she must be EFX towards the other agent. 
\end{proof}

\noindent 
\newtext{\citet{goldberg2023frontier} recently showed that it is possible to efficiently compute EFX allocations when there are only two agents for some more general valuation functions such as gross substitute or budget-additive, but the problem becomes PLS-complete for submodular functions.}

For instances with three agents, in a breakthrough paper, \citet{ChaGM20} presented a procedure that computes an EFX allocation in pseudo-polynomial time. The original proof of \citeauthor{ChaGM20} is quite involved and consists of careful and complex case analysis; a slightly simpler proof was recently given by \citet{journals/corr/abs-2205-07638}. Computing EFX allocations for three agents in polynomial time is still an open problem. 

\subsubsection{Bi-valued Valuations}
For bi-valued instances, in which there are only two distinct possible values that an agent may have for the goods, \citet{ABFHV21} showed that EFX allocations exist and can be efficiently computed for any number of agents. 
Later, \citet{GargM21} showed that this is possible even in conjunction with PO. In fact, \citet{ABFHV21} also demonstrated an interesting connection between EFX and MNW allocations for bi-valued instances, by showing that MNW allocations always satisfy EFX. A similar result was later shown by \cite{BabaiEF21} for general valuations with binary (i.e., $0$ or $1$) marginals.

\subsubsection{Discussion}
It is instructive to mention that the original definition of EFX due to \citet{CaragiannisKMPS19} was a bit weaker and required the removed good to be positively-valued. In the case of binary valuations, a special case of bi-valued instances, the distinction between the two versions of EFX does make a difference, but in more general settings the existence and computation of the weak version can be reduced to the existence and computation of the stronger version; see \citet{ABFHV21} for a related discussion. It is easy to see that EF implies EFX, which implies weak EFX, which in turn implies EF1. 

The aforementioned results are positive first steps towards showing the existence of EFX allocations, but a general positive (or negative) result remains elusive. This brings us to our second open problem, which is one of the most important open problems in discrete fair division.

\begin{open}\label{open:EFX}
Do EFX allocations exist for instances with $n \geq 4$ agents and unrestricted additive valuations?
\end{open} 

\subsection{Relaxations of EFX}
Instead of focusing on exact EFX allocations, a growing line of work has taken a different approach by aiming to compute allocations that are \emph{approximately} EFX, for different notions of approximation. The first such notion is in terms of multiplicative approximations to the values obtained by the agents.

\begin{defn}[$\alpha$-EFX]
Let $\alpha \in (0,1]$. An allocation $A$ is $\alpha$-EFX if, for every pair of agents $i,j \in N$, $v_i(A_i) \geq \alpha \cdot v_i(A_j \setminus \{g\}$) for any $g\in A_j$. 
\end{defn}

\noindent 
\citet{PR18} were the first to pursue this, showing that 1/2-EFX allocations always exist, even for subadditive valuation functions; later, \citet{ChanCLW19} showed that computing such allocations can be done in polynomial time. 
Indeed, it can be shown that for agents with additive valuation functions, the Envy-Cycle Elimination algorithm (Algorithm~\ref{alg:envy-cycle}) computes a $1/2$-EFX allocation, as long as we break ties in favor of empty bundles when picking an unenvied agent.

\begin{theorem}[\citep{ChanCLW19}] \label{thm:1/2-EFX}
The Envy-Cycle Elimination algorithm computes a $1/2$-EFX allocation.
\end{theorem}

\begin{proof}
Consider any two agents $i$ and $j$. We must show that $v_i(A_i) \geq 1/2\cdot v_i(A_j\setminus \{g\})$ for any $g\in A_j$.
The statement holds trivially if $|A_j| = 1$. So in the following, we assume that $|A_j| \geq 2$.

Let $g^*$ be the last good allocated to agent $j$. Since we break ties in favor of empty bundles when picking an unenvied agent, it must be the case that $A_i \neq \emptyset$, as otherwise $g^*$ would not be allocated to $j$, but to $i$.
Since $A_i \neq \emptyset$, there must be a round before $g^*$ is allocated during which agent $i$ gets to pick an item.
Such a round exists because empty bundles will not be involved in the envy cycle elimination procedure. Moreover, while items in $A_i$ might be swapped a few times during the envy cycle elimination, the value of $i$ for $A_i$ is non-decreasing. Therefore, $v_i(A_i) \geq v_i(g^*)$. By the design of the algorithm, we have that 
$v_i(A_i) \geq v_i(A_j\setminus \{g^*\})$.
Combining the two lower bounds, we have
\begin{align*}
    2\cdot v_i(A_i) \geq v_i(A_j) \geq v_i(A_j\setminus\{g\}), \quad \forall g\in A_j,
\end{align*}
and thus the allocation is $1/2$-EFX.
\end{proof}

The approximation ratio for the additive case was further improved by \citet{ANM2019} to $\phi-1 \approx 0.618$ by combining Round-Robin and Envy-Cycle Elimination with some appropriate pre-processing. To this end, we have the following open question: 

\begin{open}\label{open:alpha-EFX}
What is the best possible $\alpha$ for which $\alpha$-EFX allocations exist? 
\end{open}

A positive answer to Open Problem~\ref{open:EFX} would establish that $\alpha=1$ in Open Problem~\ref{open:alpha-EFX}, but a negative answer would make the latter open problem very meaningful in its own right. Additionally, as  is the case for all of these notions, the next natural question is whether existence can be paired with polynomial-time algorithms for finding such allocations, or whether some kind of computational hardness can be proven.

\smallskip

Another recent approach is that of relaxing the requirement to allocate \emph{all} available goods. Clearly, if done without any constraints, this makes the problem trivial: simply leaving all goods unallocated results in an envy-free allocation. However, the objective here is to only leave ``a few'' goods unallocated (e.g., donate them to charity instead), or remove some goods without affecting the maximum possible Nash social welfare by ``too much''. On this front, \citet{caragiannis2019charity} showed that it is possible to compute an EFX allocation of a subset of the goods, the Nash welfare of which is at least half of the maximum Nash welfare on the original set. \citet{CKMS20} presented an algorithm that computes a partial EFX allocation, but the number of unallocated goods is at most $n-1$, and no agent prefers the set of these goods to her own bundle. 
The latter result was recently improved by \citet{aaai/BergerCFF22} who showed that the unallocated goods can be decreased to $n-2$ in general, and to just one for the case of four agents. Finally, \citet{ChaudhuryGMMM21} showed that a $(1-\varepsilon)$-EFX allocation with  a sublinear number of unallocated goods and high Nash welfare can be computed in polynomial time for every constant $\varepsilon\in (0,0.5]$. This motivates the next open problem.

\begin{open}
Is it possible to achieve an exact EFX allocation by donating a sublinear number of goods?
\end{open} 

\section{Maximin Share Fairness (MMS)}
\label{sec:MMS}

As we have mentioned, although the fairness notion of MMS significantly weakens the requirement of PROP, MMS allocations are still not guaranteed to exist. Therefore, the majority of the literature has focused on approximate MMS allocations.

\begin{defn}[$\alpha$-MMS]\label{def:mms:app}
Let $\alpha \in (0,1]$. An allocation $A$ is $\alpha$-MMS if $v_i(A_i)\geq \alpha\cdot \mu^n_i(M)$ for all $i\in N$. 
\end{defn}

\noindent
Before introducing the state-of-the-art results, we first state some properties about MMS allocations that are helpful for designing approximation algorithms.
The first property is the monotonicity of the MMS values; specifically, if we arbitrarily remove one agent and one good from the instance, the MMS value of each of the remaining agents in the remaining instance does not decrease. 
This property was first formally stated by \citet{BouveretL16} (as the $k=1$ case of Lemma 1 therein) and since then several variants have appeared in the literature. Here, we state the version of \citet{AMNS17}.

\begin{lemma}[\citep{BouveretL16,AMNS17}]
\label{lem:mms:monotonicity}
For any agent $i$ and any good $g$, $\mu_i^{n-1}(M\setminus\{g\}) \ge \mu_i^{n}(M)$.
\end{lemma}

\begin{proof}
Let $(A_1,\ldots,A_n)$ be an $n$-partition of $M$ that achieves $\mu_i^{n}(M)$, i.e., $v_i(A_j) \ge \mu_i^{n}(M)$ for all $j \in [n]$. Without loss of generality, suppose $g \in A_n$.
We arbitrarily reallocate the goods in $A_n \setminus \{g\}$ (if any) to the sets $A_1,\ldots,A_{n-1}$, and denote the resulting allocation by $(B_1,\ldots,B_{n-1})$; clearly, this is a partition of $M\setminus\{g\}$ into $n-1$ bundles.
Since $v_i(B_j) \ge v_i(A_j) \ge \mu_i^{n}(M)$ for all $j \in [n-1]$, it follows that $\mu_i^{n-1}(M\setminus\{g\}) \ge \mu_i^{n}(M)$.
\end{proof}

\noindent
Lemma~\ref{lem:mms:monotonicity} essentially suggests that if there is a large good $g$ that satisfies the desired approximation of MMS value for some agent $i$, we can simply assign $g$ to $i$ and finalize the bundle $A_i = \{g\}$ without hurting the remaining agents in the reduced instance.

\subsection{Computing $1/2$-MMS Allocations}
Here, we consider three simple algorithms for computing $1/2$-MMS allocations in polynomial time. Better approximations can be obtained by variants and extensions of the techniques employed by these algorithms; see Section~\ref{sec:MMS-better-approx}.

The monotonicity of the MMS value (Lemma~\ref{lem:mms:monotonicity}) allows us to focus only on the case when the agents have sufficiently small values for the goods, in particular, when $v_i(g) \le \mu_i^n/2$ for every $i \in N$ and $g \in M$. Specifically, as long as there is an agent $i$ and a good $g$ such that $v_i(g) > \mu_i^n/2$, we can allocate $g$ to $i$ so that $i$ obtains half of her MMS value, and never consider $i$ or $g$ again. If we can guarantee $1/2$-MMS for the remaining agents in the reduced instance, the same guarantee must also hold in the original instance by Lemma~\ref{lem:mms:monotonicity}. One way to achieve this is by using Round-Robin (see Algorithm~\ref{alg:round-robin}) \citep{AMNS17}. 

\begin{lemma}
If $v_i(g) \le 1/2\cdot \mu_i^n$ for every $i \in N$ and $g \in M$, the Round-Robin algorithm returns an allocation that is $1/2$-MMS.
\end{lemma}

\begin{proof}
Let $A=(A_1,\ldots,A_n)$ be the allocation returned by the Round-Robin algorithm. 
If $m < n$, then $\mu_i^n = 0$ for every $i\in N$ and thus any allocation is MMS. So, assume that $m \ge n$.
By the design of the algorithm, for each agent $i$, the item selected in any round must be at least as good as the items allocated to any agent in subsequent rounds.
For any $i \in N$, let $g_i$ be the good selected by agent $i$ in the first round. 
Now let $i$ be an arbitrary agent.
Then, $v_i(A_i) \ge v_i(A_j \setminus\{g_j\})$ for any $j \in N$.
By summing these inequalities over $j\in N$, we obtain 
\begin{align*}
    n\cdot v_i(A_i) \ge v_i(M) - \sum_{j\in N} v_i(g_j) \ge n\cdot \mu_i^n - \frac{n}{2}\cdot \mu_i^n = \frac{n}{2}\cdot \mu_i^n,
\end{align*}
where the second inequality follows since $v_i(M) \ge n\cdot \mu_i^n$ and $v_i(g_j) \le \mu_i^n/2$ for all $j\in N$. Therefore, for any agent $i$, we have that $v_i(A_i) \geq \mu_i^n/2$, that is, she obtains at least half of her MMS value and the allocation is $1/2$-MMS.
\end{proof}

It is also not hard to see that the Envy Cycle Elimination algorithm (Algorithm~\ref{alg:envy-cycle}) returns an allocation that is $1/2$-MMS (in addition to being EF1 and $1/2$-EFX).

\begin{theorem}
The Envy-Cycle Elimination algorithm computes a $1/2$-MMS allocation.
\end{theorem}

\begin{proof}
Consider any agent $i\in N$.
As we have shown in the proof of Theorem~\ref{thm:1/2-EFX}, for every agent $j\neq i$, we have either $|A_j| = 1$ or $v_i(A_j)\leq 2\cdot v_i(A_i)$.
Let $N_1 = \{ j\in N\setminus\{i\}: |A_j|=1 \}$ be the set of agents other than $i$ that are allocated a single item, and $n_1 = |N_1|$.
By Lemma~\ref{lem:mms:monotonicity}, we have $\mu_i^{n-n_1}(M\setminus(\cup_{j\in N_1} A_j)) \ge \mu_i^{n}(M)$.
In other words, the MMS value of agent $i$ in the instance with set of agents $N_2 = N\setminus N_1$ and set of items $M\setminus(\cup_{j\in N_1} A_j)$ is at least as large as $\mu_i^n(M)$.
Since $v_i(A_j) \leq 2\cdot v_i(A_i)$ for every $j \in N_2$, we have that $v_i(\cup_{j\in N_2} A_j) \leq 2\cdot |N_2|\cdot v_i(A_i)$.
Reordering the inequality gives
\begin{align*}
    v_i(A_i) \geq \frac{1}{2}\cdot \frac{v_i(\cup_{j\in N_2} A_j)}{|N_2|} \geq \frac{1}{2}\cdot \mu_i^{n-n_1}(M\setminus(\cup_{j\in N_1} A_j)) \geq \frac{1}{2}\cdot \mu_i^n(M),
\end{align*}
which implies that the allocation is $1/2$-MMS for every agent $i\in N$.
\end{proof}

The third algorithm for computing a $1/2$-MMS allocation relies on a commonly used technique known as {\em bag-filling} \citep{ghodsi2021improvement}. Without loss of generality, and for simplicity, we assume that the MMS values of all agents are normalized to 1: $\mu_i^n = 1$ for every $i\in N$. As in the case of the first algorithm, due to Lemma~\ref{lem:mms:monotonicity}, we can focus on the case where $v_i(g) \le 1/2 = \mu_i^n/2$ for every $i \in N$ and $g\in M$.
We start with an empty bag and keep adding goods one by one into the bag until there is an agent with value at least $1/2$ for the goods in the bag. Then, we allocate all the goods in the bag to this agent. This process is repeated for the remaining agents and goods. When a single agent remains, she is given all the available goods. See Algorithm~\ref{alg:bag-filling}.

\begin{algorithm}[ht]
\caption{Bag-Filling Algorithm}
\label{alg:bag-filling}
\textbf{Input: }{A fair allocation instance $I = (N,M,v)$ with $|N| \ge 2$ and $\mu_i^{n}(M)=1$ for all $i\in N$}\\
\textbf{Output: }{Allocation $A=(A_1,\ldots,A_n)$} \\
\For{each agent $i \in N$}{
    $A_i \gets \varnothing$\; 
}
$\BAG \gets \varnothing$\;

\While{$|N| > 1$ and there exist $i\in N$ and $g\in M$ such that $v_i(g) > 1/2$}{
    $A_i \gets \{g\}$\;
    $N \gets N\setminus\{i\}$\;
    $M \gets M\setminus\{g\}$\;
}

\While{$|N| > 1$}{
    \While{$v_i(\BAG) < 1/2$ for every $i\in N$}{
        Select an arbitrary good $g \in M$\; 
        Set $\BAG \gets \BAG \cup \{g\}$\; 
        Set $M \gets M\setminus\{g\}$\;
    }
    Let $i^* \in N$ be any agent with $v_{i^*}(\BAG) \ge 1/2$\;
    $A_{i^*} \gets \BAG$\; 
    $N \gets N\setminus\{i^*\}$\;
    $\BAG \gets \varnothing$\; 
}  

Let $i$ be the remaining agent in $N$\; 
$A_i \gets M$\;
\end{algorithm}

\begin{theorem}
The Bag-Filling Algorithm (Algorithm \ref{alg:bag-filling}) returns a $1/2$-MMS allocation.
\end{theorem}

\begin{proof}
First, recall that any agent has a total value of at least $n$ for the all the goods due to the normalization. In each round of the algorithm, the bag is allocated to an agent that has value at least $1/2$ for the goods in the bag, and thus this agent obtains at least half of her MMS value. The value of each of the other agents for the given bag is smaller than her MMS value of $1$: before the last good is added to the bag their value is smaller than $1/2$ and each good (including the last one added to the bag) has value at most $1/2$. Consequently, when an agent takes away a bag in some round, each of the remaining $n' < n$ agents has total value at least $n'$ for the remaining goods. The argument can be repeated until there is a single agent. Then she gets all the remaining goods of value at least $1$, i.e., her whole MMS value.
\end{proof}

\subsection{Better Approximations}\label{sec:MMS-better-approx}
\citet{KurokawaPW18} were the first to go beyond $1/2$ by showing how to find $2/3$-MMS allocations, albeit not in polynomial time. \citet{AMNS17} matched this guarantee in polynomial time, as did \citet{BarmanK20} with a much simpler algorithm, which relies on an important property that was first proved by \citet{BouveretL16} for exact MMS allocations and then generalized by \citet{BarmanK20} for approximate MMS allocations. This property is very important and suggests that the hardest case is when the agents have the same order of preference over the goods; in other words, there exists a worst-case instance that is ordered.

Given any instance $I=(N,M,v)$ which may not be ordered, we can construct a corresponding ordered instance $I'=(N,M,v')$. 
Based on an ordering $\succ$ of the goods, for each agent $i\in N$, we find a permutation $\sigma_i: M \to M$ such that $v_i(\sigma_i(g)) \ge v_i(\sigma_i(g'))$ for every $g \succ g'$.
Then, we define a new valuation function $v'_i$ with $v'_i(g) = v_i(\sigma_i(g))$ for every $g\in M$; hence, the new value of $i$ for good $g$ that is ranked $\ell$-th in $\succ$ equals the $\ell$-th highest value of $i$ in the original valuation. Given any allocation $A'$ for the ordered instance $I'$, we can obtain an allocation $A$ for the original instance $I$ such that every agent has value for $A_i$ at least as much as for $A'_i$. Hence, if $I'$ admits an MMS allocation then $I$ admits an MMS allocation as well, and the same holds for approximate MMS allocations.

\begin{lemma}[\citep{BouveretL16,BarmanK20}]
\label{lem:reduction-to-ordered}
Let $I=(N,M,v)$ be any instance, $I'$ be its corresponding instance, and $A'=(A_1',\ldots,A_n')$ be an allocation for $I'$ such that $v'_i(A'_i) \ge \alpha_i$ for every $i \in N$. Then, we can construct in polynomial time an allocation $A=(A_1,\ldots,A_n)$ for $I$ such that $v_i(A_i) \ge \alpha_i$ for every $i \in N$.
\end{lemma}

\citet{BarmanK20} showed that the Envy-Cycle Elimination algorithm actually returns a $2/3$-MMS allocation for ordered instances, and thus a $2/3$-MMS fair allocation can be found in polynomial time for every instance due to Lemma~\ref{lem:reduction-to-ordered}. The barrier of $2/3$ was broken by \citet{ghodsi2021improvement} who designed an elaborate $(3/4-\epsilon)$-approximation algorithm. A simpler algorithm with a slightly improved approximation guarantee of $3/4 + 1/(12n)$ was then proposed by \citet{GT21}. On the negative side, \citet{FST21} recently showed that it is impossible to achieve an approximation bound better than $39/40$, even for the case of three agents.

\begin{open}
Is it possible to improve upon the bound of $3/4 + 1/(12n)$ for additive valuations? Is there a stronger inapproximability bound than $39/40$?
\end{open}

\subsection{Restricted Instances}

As expected, by restricting the number of agents or the space of the valuation functions, one can get stronger results.

\paragraph{MMS for four or fewer agents.} 
When there are only two agents, a simple cut-and-choose-type protocol always produces an MMS allocation. Specifically, the first agent partitions the set of goods as equally as possible (thus the worst set has value equal to her maximin share) and the second agent decides who gets each of these sets. As suggested above, the first step is computationally hard but producing a $(1-\epsilon)$-MMS allocation in polynomial time is still possible.
Even though in the general case, the algorithm of \citet{KurokawaPW18} guarantees a $2/3$-approximation, for three or four agents it guarantees an improved $3/4$-approximation. 
The approximation factor for three agents was then improved to $7/8$ \citep{AMNS17}, later to $8/9$ \citep{gourves2019MMS} and recently to $11/12$~\citep{feige2022improved}, whereas for the case of four agents the factor was improved to $4/5$ \citep{ghodsi2021improvement}.

\paragraph{MMS for restricted valuations.} 
It follows by Definition \ref{def:mms} that MMS allocations exist for instances where all agents have identical valuation functions. As discussed above, \citet{BouveretL16} showed that, unlike with EFX, the hardest instances for MMS (among all possible instances) are the ordered instances. They also suggested a simple construction of exact MMS allocations when the valuation functions are binary. 
This result can be generalized, as MMS allocations always exist and can be computed efficiently for ternary valuation functions \citep{AMNS17} and for bi-valued valuation functions \citep{Feige22}. \citet{conf/atal/EbadianP022} showed that this is also the case when there are at most two values per agent (possibly not common across all agents) and for general instances where the value of each good is at least as much as the value of all less valuable goods combined. \newtext{Another notable class for which exact MMS allocations are known to exist is that of matroid-rank valuations \citep{barman2021existence}.}

\begin{open}
Are there other classes of structured valuations for which MMS is guaranteed to exist, e.g., when there are only a few (but more than two) possible values?
\end{open}

\section{Further Notable Fairness Notions} \label{sec:other-notions}

In this section, we discuss other notable fairness notions beyond EF, PROP, EF1, EFX and MMS.

\subsection{EFL, EFR, and Epistemic Notions}
The EFX notion was defined as a more realistic counterpart to EF1, however, as we have discussed in Section~\ref{sec:EFX}, it is still unknown if it can always be guaranteed. This has led to the definition of notions that lie in-between EF1 and EFX. \citet{BBMN18} defined the notion of {\em envy-freeness up to one less-preferred good} (EFL) according to which an agent $i$ does not envy another agent $j$ if either $A_j$ contains at most one good that $i$ values positively, or the envy of $i$ can be eliminated by the hypothetical removal of a good $g \in A_j$ such that $v_i(A_i) \geq v_i(g)$. They showed that EFL allocations always exist and can be computed using a variant of the Envy-Cycle Elimination algorithm. \citet{farhadi2021efr} defined the notion of {\em envy-freeness up to a random good} (EFR) according to which the envy of an agent $i$ towards another agent $j$ can be eliminated \emph{in expectation} after the hypothetical removal of a randomly chosen good from $A_j$. They showed that a $0.73$-EFR allocation can be computed in polynomial time.

\begin{open}
Do EFR allocations always exist?
\end{open}

\newtext{Another recent line of research has focused on {\em epistemic} fairness notions motivated by cases where the agents do not have complete knowledge about the computed allocation \citep{chen2017ignorance,aziz2018knowledge,DBLP:conf/aaai/HosseiniSVWX20,garg2022existence}. 
For example, \citet{aziz2018knowledge} considered a setting where the agents are connected via a social network and do not know the bundles allocated to the agents who are not their neighbors. They introduced the notion of {\em epistemic EF} under which an agent is satisfied only if there is a (re)distribution of the goods she is not aware of (i.e., not given to her or to her neighbors) such that the resulting allocation is EF to her. Recently, \citet{garg2022existence} proved the existence of epistemic EFX allocations for additive valuations when the underlying social network contains only isolated nodes.}

\subsection{PMMS and GMMS} 
Several variations of MMS have also been considered. \citet{CaragiannisKMPS19} defined the notion of {\em pairwise maximin share fairness} (PMMS) according to which, for every pair of agents $i$ and $j$, $i$'s value for $A_i$ must be at least as much as the maximum she could obtain by redistributing the set of goods in $A_i \cup A_j$ into two bundles and picking the worst of them. In other words, instead of requiring the maximin share guarantee to be achieved for the set of all agents, PMMS requires that it is achieved for any pair of agents. 

\begin{defn}[$\alpha$-PMMS]
Let $\alpha \in (0,1]$. An allocation $A$ is $\alpha$-PMMS if $v_i(A_i)\geq \alpha\cdot \mu^2_i(A_i \cup A_j)$ for all $i, j\in N$. When $\alpha=1$, the allocation is called PMMS.
\end{defn}

\noindent
Despite the apparent similarities in the definitions of PMMS and MMS, \cite{CaragiannisKMPS19} showed that their exact versions are actually incomparable. The main open problem here is the following.

\begin{open}
Do PMMS allocations always exist?
\end{open}

\noindent
\newtext{In fact, any PMMS allocation must be EFX, and thus showing the existence of PMMS allocations would also imply the existence of EFX allocations (Open Problem~\ref{open:EFX}).} For approximate PMMS, the best known result is $0.781$ by  \citet{Kurokawa17}. 

An even stronger notion, which implies both MMS and PMMS, is that of {\em groupwise maximin share fairness} (GMMS) defined by \citet{BBMN18}, and which requires that the maximin share guarantee is simultaneously achieved for any possible subset of agents. 

\begin{defn}[$\alpha$-GMMS]
Let $\alpha \in (0,1]$. An allocation $A$ is $\alpha$-GMMS if $v_i(A_i)\geq \alpha\cdot \GMMS_i$ for all $i \in N$, where 
$$\GMMS_i = \max_{S \subseteq N: i\in S} \mu_i^{|S|} (\bigcup_{j \in S} A_j).$$
When $\alpha=1$, the allocation is called GMMS.
\end{defn}

\noindent
\citet{BBMN18} showed that GMMS allocations exist for some restricted settings, such as when the agents have  binary or identical values. They also showed that any EFL allocation is $1/2$-GMMS, and thus such an allocation can be computed efficiently. The currently best known approximation of GMMS is $4/7$  \citep{ANM2019,CKMS20}. The implication relations between all the aforementioned notions has been used many times to show that particular algorithms have guarantees that hold for multiple notions at once. We refer the reader to the paper of \citet{ABM18} for a discussion of the relations between (approximate versions of) these notions.

\begin{open}\label{open:alpha-GMMS}
What is the best possible $\alpha$ for which $\alpha$-GMMS allocations exist? 
\end{open}

\subsection{Prop1, PropX and PropM}

A line of work has also focused on relaxations of proportionality that are similar in essence to EF1 and EFX. 
\citet{conitzer2017prop1} defined the notion of {\em proportionality up to one good} (Prop1) according to which each agent could obtain her proportional share if given one extra good. 
\cite{aziz2020propx} defined {\em proportionality up to any good} (PropX) that demands that each agent can obtain her proportional share when given the least positively-valued good among those allocated to other agents.

\begin{defn}[Prop1]
An allocation $A$ is proportional up to one good (Prop1) if, for every agent $i \in N$, there exists a good $g \in M\setminus A_i$ such that $v_i(A_i \cup \{g\})\geq v_i(M)/n$.
\end{defn}

\begin{defn}[PropX]
An allocation $A$ is proportional up to any good (PropX) if, for every agent $i \in N$ and any good $g \in M\setminus A_i$, we have $v_i(A_i \cup \{g\})\geq v_i(M)/n$.
\end{defn}

\noindent
An allocation that is Prop1 and PO always exists \citep{conitzer2017prop1} and can be computed in polynomial time \citep{Barman2019prop1}. 
PropX is rather demanding and it cannot always be guaranteed, as shown by~\citet{aziz2020propx} via the following simple example (see Table~\ref{table:example:no-PropX}). It can be verified that there must exist an agent $i$ that receives at most one item. Moreover, there exists a good $g \not\in A_i$ whose inclusion gives agent $i$ a total value of $4$, which is lower than $13/3$, the proportional share of agent $i$.

\begin{table}[ht]
    \centering
    \begin{tabular}{c|ccccc}
        & $g_1$ & $g_2$ & $g_3$  & $g_4$  & $g_5$ \\ 
        \hline
        $a_1$ & $3$ & $3$ & $3$  & $3$  & $1$ \\ 
        $a_2$ & $3$ & $3$ & $3$  & $3$  & $1$ \\ 
        $a_3$ & $3$ & $3$ & $3$  & $3$  & $1$ \\ 
        \hline
    \end{tabular}
    \caption{An instance that does not admit any PropX allocation.}
    \label{table:example:no-PropX}
\end{table}

Recently, \citet{baklanov2021propm1,baklanov2021propm2} introduced the notion of {\em proportionality up to the maximin good} (PropM) which provides a middle-ground between Prop1 and PropX.
Informally, given an allocation $A$, the value of a maximin good for agent $i$ is 
$$
d_i(A) = \max_{j \neq i} \min_{g\in A_j} v_i(g),
$$
and $A$ is called PropM if $v_i(A_i) + d_i(A) \ge v_i(M)/n$.
\citet{baklanov2021propm2} showed that a PropM allocation always exists and can be computed in polynomial time. 

\subsection{Equitability and its Relaxations}
Besides envy-freeness and proportionality, another important fairness notion is that of {\em equitability} according to which all agents must derive the same value from the bundles they are allocated. 

\begin{defn}[Equitability \citep{books/daglib/0017730}]
An allocation $A=(A_1,\dots,A_n)$ is {\em equitable} (EQ) if $v_i(A_i) = v_j(A_j)$ for every pair of agents $i,j \in N$.
\end{defn}

\noindent
It is not hard to observe that an EQ allocation may not exist, even when there is just one good and two agents. 
Similarly to EF1 and EFX, we can relax EQ to EQ1 and EQX.

\begin{defn}[EQ1]
An allocation $A=(A_1,\dots,A_n)$ is {\em equitable up to one good} (EQ1) if, for every pair of agents $i,j \in N$, it holds that $v_i(A_i) \ge v_j(A_j\setminus\{g\})$ for some $g\in A_j$.
\end{defn}

\begin{defn}[EQX]
An allocation $A=(A_1,\dots,A_n)$ is {\em equitable up to any good} (EQX) if, for every pair of agents $i,j \in N$, it holds that $v_i(A_i) \ge v_j(A_j\setminus\{g\})$ for any $g\in A_j$.
\end{defn}

\noindent 
Clearly an allocation that is EQX is also EQ1. \citet{GourvesMT14} considered EQX allocations under the alias of {\em nearly jealousy-free}, and showed that such allocations (and thus also EQ1 allocations) always exist and can be computed in polynomial time. 

\section{Beyond Fairness: Efficiency and Incentives}
\label{sec:efficiency&truthfulness}

Here we consider other interesting directions like the relation between fairness and efficiency, or the possibility of achieving fairness when the agents are strategic. 

\subsection{Fair and Pareto Optimal Allocations} \label{sec:fair-PO}
There is a significant line of work that considers the question of whether it is possible to simultaneously achieve fairness and \emph{efficiency}. A common type of efficiency is that of PO (see Definition~\ref{def:po}), which, as we already discussed, can be guaranteed in conjunction with some fairness notions, like EF1 and Prop1. 
In particular, \citet{CaragiannisKMPS19} showed that any allocation that maximizes the Nash welfare is EF1 and PO (Theorem~\ref{thm:MNW:EF1+PO}). However, since the computation of Nash welfare maximizing allocations is NP-hard, this result only shows the existence of EF1+PO allocations.

For the computation of allocations that are fair and PO, a commonly used technique in the literature is to exploit the connection between fair division and {\em market equilibria} in {\em Fisher markets} \citep{Budish11}, which has been widely used in the design of approximation algorithms for maximizing the Nash welfare~\citep{journals/siamcomp/ColeG18,conf/sigecom/0001DGJMVY17}.
We briefly introduce this framework.

Similarly to a fair division instance, a Fisher market consists of a set of agents $N$, a set of goods $M$, and each agent $i \in N$ has a value $v_i(g)$ for every good $g \in M$. In contrast to discrete fair division, each agent $i$ also has a {\em budget} $b_i \geq 0$, and the goods can be fractionally allocated. A {\em fractional allocation} is a tuple of vectors $X=(x_i)_{i \in N}$, where $x_i = (x_{i1},\ldots,x_{im})$ is the allocation of agent $i$; in particular, $x_{ig} \in [0,1]$ denotes the fraction of $g$ that is allocated to $i$, such that $\sum_{i \in N} x_{ig} \leq 1$ for every $g \in M$. We will be focusing on Fisher markets with linear utilities, where, given a fractional allocation $X$, the value of agent $i$ value is $v_i(x_i) = \sum_{g\in M} x_{ig} \cdot v_i(g)$. 
A {\em market outcome} is a tuple $(X,p)$, where $X$ is a fractional allocation and $p=(p_1,\dots,p_m)$ defines the price $p_j$ for item $j$. 
An outcome $(X,p)$ is a {\em market equilibrium} if it satisfies the following conditions:
\begin{enumerate}
    \item The market clears (all items are completely allocated): 
    $\sum_{i\in N} x_{ig} = 1$ for every $g\in M$ with $p_g >0$;
    \item The budgets of all agents are exhausted: 
    $\sum_{g\in M} p_g \cdot x_{ig} = b_i$ for every $i\in N$;
    \item The agents only spend money on the goods with maximum value per unit of money spent (bang-per-buck):
    For any $g \in M$, $x_{ig}>0$ implies $\frac{v_i(g)}{p_g} \ge \frac{v_i(g')}{p_{g'}}$ for every $g'\in M$.
\end{enumerate}

An allocation is {\em fractionally Pareto optimal} (fPO) if it is not Pareto dominated by any fractional allocation.
Clearly, an fPO allocation is also PO.
It is known that a market equilibrium leads to an fPO allocation.
Moreover, when all agents have the same budget (that is, we have a {\em competitive equilibrium from equal incomes} (CEEI)), the allocation is EF \citep{Varian74}.
Therefore, if a market equilibrium gives an integral allocation, this allocation is EF and fPO. However, this rarely happens.

Many works have used Fisher markets as a proxy. For example, \citet{BarmanKV18} and \citet{Barman2019prop1} compute allocations by perturbing the budgets or the values in the market so that they admit market equilibria with integral allocations. Based on this, \citet{BarmanKV18} designed a pseudo-polynomial time algorithm to compute an allocation that is EF1 and PO, and \citet{Barman2019prop1} designed a strongly polynomial time algorithm to compute an allocation that is PROP1 and PO. \citet{garg2022computing} improved the result of \citet{BarmanKV18} by designing a pseudo-polynomial time algorithm to compute an EF1 and fPO allocation.
\citet{GargM21} showed that in bi-valued instances, an EFX and PO allocation can be computed in polynomial-time, but in instances with three distinct values, EFX and PO are incompatible.
\citet{conf/ijcai/FreemanSVX19} showed that, if all values are strictly positive, there always exists a EQX and PO allocation. On the negative side, using Example~\ref{ex:no-EQ1&PO}, they showed that there do not exist EQ1 and PO allocations when the values can be zero. \citet{garg2022computing} strengthened this result by providing a pseudo-polynomial time algorithm to compute an EQ1 and PO allocation.

\begin{ex}\label{ex:no-EQ1&PO}
\begin{table}[ht]
    \centering
    \begin{tabular}{c|cccccc}
        & $g_1$ & $g_2$ & $g_3$  & $g_4$  & $g_5$ & $g_6$ \\ 
        \hline
        $a_1$ & $1$ & $1$ & $1$  & $0$  & $0$ & $0$ \\ 
        $a_2$ & $0$ & $0$ & $0$  & $1$  & $1$ & $1$ \\ 
        $a_3$ & $0$ & $0$ & $0$  & $1$  & $1$ & $1$ \\ 
        \hline
    \end{tabular}
    \caption{An instance that does not admit any EQ1 and PO allocation~\citep{conf/ijcai/FreemanSVX19}.}
    \label{table:example:no-EQ1&PO}
\end{table}

Consider the instance given in Table~\ref{table:example:no-EQ1&PO}.
In any PO allocation, $\{g_1,g_2,g_3\}$ must be assigned to agent $a_1$, and the remaining goods $\{g_4,g_5,g_6\}$ are shared between agents $a_2$ and $a_3$. Consequently, one of $a_2$ and $a_3$ is allocated at most one good, which violates the requirement of EQ1 with agent $a_1$. 
\hfill $\qed$
\end{ex}

\paragraph{\newtext{Approximating the Nash welfare.}}
\newtext{Since MNW allocations are both EF1 and PO (see Section~\ref{sec:EF1}), it is worth mentioning at this point that the problem of efficiently computing approximate maximum Nash welfare allocations has been among the most active ones in computational social choice. \citet{garg2017satiation} showed that approximating the Nash welfare within a ratio better than $1.069$ is NP-hard, even for additive valuation functions with only four values; in fact, it is APX-hard to approximate the Nash welfare even when there only two values~\citep{conf/aaai/AkramiC0MSSVVW22}. For additive valuations, the best known approximation is $e^{1/e}+\epsilon\approx 1.45$ \citep{BarmanKV18}. Closing the gap is an interesting open question. It should be noted, however, that approximate Nash welfare allocations are not known to be even approximately EF1.}

\newtext{When all agents have submodular valuation functions (see Section~\ref{ssec:general-valuations} for a formal definition of submodular, XOS and subadditive functions), the best known approximation of Nash welfare is $6+\epsilon$~\citep{journals/corr/GargHLVV22} and the best known lower bound is $1.5819$~\citep{conf/soda/GargKK20}.
Under the value oracle model, there is an assymptotically tight approximation ratio of $\Theta(n)$ for both XOS and subadditive functions~\citep{conf/esa/BarmanBKS20,conf/aaai/ChaudhuryGM21}. With stronger oracle assumptions, a sublinear approximation ratio is possible for XOS functions~\citep{journals/corr/BarmanKKN21}. Whether sublinear approximation ratios are possible for subadditive functions with stronger oracle assumption remains unknown.}

\paragraph{\newtext{Egalitarian welfare.}}
\newtext{Another common welfare function that is considered as a fairness criterion itself is the {\em egalitarian welfare}, defined as the minimum utility among all agents (see, e.g., \citep{DBLP:journals/sigecom/BezakovaD05}). Aiming to maximize the egalitarian welfare is a known APX-hard problem known as the {\em Santa Claus} problem in the literature (e.g., see \citep{bansal2006santa,feige2008allocations,lenstra1990approximation}). While egalitarian welfare maximizing allocations are clearly PO, they may not achieve any constant approximation of important relaxations of EF and Prop, such as EF1 and MMS.}

\subsection{Price of Fairness}
Another natural goal is to (approximately) maximize some objective function of the values of the agents, such as the {\em social welfare}, i.e., the total value of the agents for the goods they receive. 
Formally, the social welfare of an allocation $A$ is $\sw(A) = \sum_{i\in N}v_i(A_i)$.
\citet{bertsimas2011efficiency} and \citet{caragiannis2012efficiency} defined the {\em Price of Fairness}, a measure which, similarly to the approximation ratio in worst-case analysis, measures the deterioration of the objective due to the fairness requirement (which may refer to any fairness notion).
Given an instance $I$, we denote by $\opt(I)$ the maximum social welfare over all allocations of $I$. 
Let $F$ be a fairness criterion (such as EF1 and MMS) and $F(I)$ be the set of all allocations satisfying $F$.

\begin{defn}[Price of Fairness (PoF)]
	The price of fairness with respect to fairness criterion $F$ is
	$$
	\PoF(F) = \sup\limits_{I} \min\limits_{A \in F(I)} \frac{\opt(I)}{\sw(A)}.
	$$
\end{defn}
\noindent
If no fair allocation can achieve non-zero welfare, the price of fairness is infinite (unbounded).
The price of fairness with respect to fairness criterion $F$ is also called the {\em price of $F$}. \citet{BeiLMS21} were the first to consider the price of fairness with respect to relaxations of EF. For instances with two agents, using
Example~\ref{ex:PoS-EF1}, they showed that the price of EF1 is at least $8/7$.

\begin{ex}\label{ex:PoS-EF1}
Consider the instance given in Table~\ref{table:example:PoEF1:2}, where $\varepsilon > 0$ is arbitrarily small.

\begin{table}[ht]
    \centering
\begin{tabular}{c|ccc}
        & $g_1$ & $g_2$ & $g_3$ \\\hline
$a_1$   &   $1/3 - 2\varepsilon$  & $1/3 + \varepsilon$     & $1/3 + \varepsilon$     \\
$a_2$   &   0   & $1/2$    & $1/2$     \\\hline
\end{tabular}
\caption{A lower bound instance for the price of EF1 for two agents.}
\label{table:example:PoEF1:2}
\end{table}

\noindent
The optimal social welfare in this instance is $4/3-2\varepsilon$, achieved by allocating $g_1$ to $a_1$ and other two goods to $a_2$. However, in any EF1 allocation the last two goods cannot both be given to $a_2$. Hence, the social welfare of an EF1 allocation is at most $(1/3 - 2\varepsilon) +(1/3 +\varepsilon) + 1/2 = 1/6-\varepsilon$. Taking $\varepsilon \to 0$, the price of EF1 is at least $8/7$.
\hfill $\qed$
\end{ex}

An upper bound on the price of EF1 follows by a variant of the Adjusted-Winner algorithm \citep{books/daglib/0017730}.
The idea is to sort the goods according to the ratios between the values that they yield to the two agents:
\begin{align*}
     \frac{v_{1}(g_1)}{v_{2}(g_1)} \ge \frac{v_{1}(g_2)}{v_{2}(g_2)} \ge \ldots \ge \frac{v_{1}(g_m)}{v_{2}(g_m)}.
\end{align*} 
Intuitively, the first goods according to this ordering are more valuable to the first agent while the last ones are more valuable to the second agent.
\citet{BeiLMS21} proved that, if a minimal set of consecutive goods from left is allocated to the first agent such that the allocation is EF1 for her and the remaining goods are given to agent 2, then the allocation is EF1 with social welfare that is within a $\sqrt{3}/2$-fraction of the optimal social welfare.

For arbitrary number of agents, \citet{BeiLMS21} showed that a lower bound of $\Omega(\sqrt{n})$ on the price of EF1; \citet{BarmanB020} managed to give a tight upper bound of $O(\sqrt{n})$.
Furthermore, \citet{BarmanB020} showed tight bounds for other fairness notions, in particular, for $1/2$-MMS and Prop1. \citet{halpern2021limited} showed tight bounds on the price of EF1 and of approximate MMS under the constraint of having only ordinal information about the agent values, a typical assumption made in the context of \emph{distortion} in social choice~\citep{distortion-survey}.

\subsection{Fair Division with Strategic Agents}
\label{subsec:strategic}
Most of the papers mentioned so far, studied the discrete fair division problem from an algorithmic perspective under the assumption that the agents are non-strategic. In the \emph{strategic} version of the problem, an agent may intentionally misreport how she values the goods in order to end up with a better bundle. This introduces an additional layer of difficulty, as the goal is to produce fair allocations \emph{according to the true values} of the agents, while their declarations might be far from being the truth. This version of the problem has been considered mostly from a mechanism design \emph{without monetary transfers} perspective, in which the utility of an agent is defined as her (true) value for her bundle. 
    
A first direction was the design of \emph{truthful} mechanisms that are also fair. 
A mechanism is truthful if no agent has an incentive to lie, i.e., no matter what values are reported by the other agents, reporting her true values always brings an agent value no smaller than what she obtains when reporting false ones.
\citet{CKKK09} showed that no truthful mechanism for two agents and two goods can always output allocations of minimum envy. \citet{ABM16} revisited the problem for the case of two agents and any number of goods, and showed that no truthful mechanism can always output $\alpha$-MMS allocations, for $\alpha > 2/m$. 
A characterization of truthful mechanisms for two agents, showing that truthfulness and fairness are incompatible (in the sense that there is no truthful mechanism with bounded fairness guarantees under \emph{any} meaningful fairness notion) was provided by \cite{ABCM17}. This impossibility, however, does not apply to restricted cases. For binary valuations, maximizing the Nash welfare results in a polynomial-time truthful mechanism that outputs EF1 and PO allocations, see \citep{0002PP020, bogomolnaia2004random, ABFHV21}. \citet{BabaiEF21} showed an analogous result with respect to MMS. In fact, \citeauthor{BabaiEF21} also showed that for the submodular analog of binary valuations, there is a truthful mechanism that always outputs EFX allocations.
    
The aforementioned impossibility results led to a different direction, where the focus was shifted to the stable states of non-truthful mechanisms. In particular, \cite{ABFLLR21} studied mechanisms that always have \emph{pure Nash equilibria}, and showed that every allocation that corresponds to an equilibrium of Round-Robin is EF1 with respect to the (unknown) true values of the agents. Qualitatively similar results can be obtained for approximate pure Nash equilibria, even for agents with submodular valuation functions \citep{ABLLR23}. 
\citet{conf/ecai/BouveretL14} and \citet{conf/aaai/AzizBLM17} studied the strategic setting for general sequential allocation algorithms. 

\begin{open}
Are there mechanisms that always have pure Nash equilibrium allocations with stronger guarantees than EF1?
\end{open}

Another direction to escape the impossibility is to relax the requirements of truthfulness to {\em not obviously manipulability} \citep{journals/corr/abs-2206-11143}. A mechanism is not obviously manipulable (NOM) if no agent can increase her best- and worst-case value by lying. Fortunately, Round-Robin algorithm is NOM, and thus EF1 and NOM are compatible. Furthermore, \citet{journals/corr/abs-2206-11143} showed that we can achieve EF1, PO and NOM simultaneously. 

\section{Different Settings}
\label{sec:other_settings}

In this last section, we briefly discuss further meaningful discrete fair division settings.

\subsection{Limited Information}
There is also an increasing interest in settings with partial information, and particularly when only ordinal information is available. Even though the agents have cardinal values for the goods, the algorithm may only have access to each agent's ranking over the goods (induced by the values in non-increasing order). 
Given only the ordinal preferences of the agents, we cannot run algorithms such as Envy-Cycle Elimination which need to calculate and compare the values of the agents. 
Instead, sequential-picking algorithms become more powerful as they do not directly rely on the values. 

It is not hard to obtain an EF1 allocation as the Round-Robin algorithm only needs to know which of the available goods every agent values the most in each round.  
However, it becomes harder to compute approximate MMS allocations.
Let $H_n =  \Theta(\ln n)$ be the $n$-th harmonic number.
\citet{halpern2021limited} showed that with only ordinal information, it is impossible to achieve better than $1/H_n$-MMS, while in previous work, \citet{ABM16} showed that $1/2H_n$-MMS allocations can be computed. 
\citet{conf/aaai/HosseiniSVX21} showed the existence of PO and MMS or EFX allocations when agents have lexicographic preferences. 

Another interesting question is to investigate the query complexity of unknown valuations.
In this model, the algorithm can access the valuations by making queries to an oracle.
\citet{journals/siamdm/OhPS21} proved that $\Theta(\log m)$ queries suffice to define an algorithm that returns EF1 allocations. In general, it is an important research direction to explore how much information about the valuations of the agents is sufficient to design algorithms with strong fairness guarantees.

\subsection{General Valuations} \label{ssec:general-valuations}
Besides additive valuations, there are classes of valuation functions (such as submodular, fractional subadditive (XOS), and subadditive~\citep{conf/sigecom/Nisan00}) which capture more complex and combinatorial preferences.
A valuation function $v_i$ is 
\begin{itemize}
    \item {\em submodular} if $v_i(S) + v_i(T) \ge v_i(S \cup T) + v_i(S \cap T)$ for all $S$ and $T$;
    \item {\em subadditive} if $v_i(S)+v_i(T) \ge v_i(S \cup T)$ for all $S$ and $T$;
    \item {\em XOS} if there is a finite number of additive functions $a_1,\ldots,a_k$ such that $v_i(S) =\max_{\ell\in [k]} a_\ell(S)$.
\end{itemize}
It is known that any additive function is also submodular, any submodular function is XOS, and any XOS function is subadditive; all these implications are strict. 

The Envy-Cycle Elimination algorithm returns an EF1 allocation even when the valuation functions of the agents are as general as possible~\citep{LMMS04}.
However, for MMS allocations, if the functions are not restricted, no approximation guarantee can be achieved as shown in Example~\ref{ex:MMS-general}.

\begin{ex}\label{ex:MMS-general}
\begin{table}[h]
    \centering
    \begin{tabular}{c|ccccc}
                & $S_1 = \{g_1,g_2\}$ & $S_2 = \{g_3,g_4\}$ & $S_3 = \{g_1,g_3\}$ & $S_4 = \{g_2,g_4\}$ & $S \neq S_1,S_2,S_3,S_4$\\
        \hline
        $a_1$   & 1     & 1     & 0     & 0     & 0  \\
        $a_2$   & 0     & 0     & 1     & 1     & 0  \\
        \hline
    \end{tabular}
    \caption{A hard instance with general valuations.}
    \label{table:variants:general}
\end{table}

Consider the instance given in Table~\ref{table:variants:general}~\citep{BouveretL16,GhodsiHSSY22}.
It is not hard to check that the MMS values of the two agents are $\mu_1^2 = \mu_2^2 = 1$. But, in any allocation there is at least one agent with value $0$, leading to an unbounded approximation ratio with respect to MMS.
\hfill $\qed$
\end{ex}

Therefore, most of the literature has focused on more structured valuations.
\citet{BarmanK20} first proved that the Round-Robin algorithm returns a $0.21$-approximate MMS fair allocation for submodular valuations. The guarantee was improved to $1/3$ by \citet{GhodsiHSSY22}, who complemented it with an upper bound of $3/4$.
For XOS valuations, \citet{GhodsiHSSY22} proved the existence of $1/5$-MMS allocations, and designed a polynomial-time algorithm to compute $1/8$-MMS allocations.  The existence result was improved to $1/4.6$ by \citet{seddighin2022improved}.
On the negative side, \citet{GhodsiHSSY22} showed that no algorithm can do better than $1/2$-approximation for all XOS valuations.
For subadditive valuations, \citet{seddighin2022improved} proved the existence of a $\Omega({1}/(\log n \log \log n))$-MMS allocation, which exponentially improves the $\Omega({1}/{\log m})$ guarantee that was shown by \citet{GhodsiHSSY22}.
We summarize the best known results for approximate MMS allocations in Table~\ref{tab:general-valuations:MMS}.

\begin{table}[ht]
    \centering
    \begin{tabular}{c|cc}
                & Lower bound       & Upper bound \\\hline
    Additive    & $3/4 + 1/(12n)$    & $39/40$ \\
    Submodular  & $1/3$ & $3/4$\\   
    XOS         & $1/4.6$ & $1/2$ \\
    Subadditive & $\Omega\left(1/\left(\log n \log \log n\right)\right)$  & $1/2$ \\\hline
    \end{tabular}
    \caption{Best known approximations of MMS for general valuations based on the works of \citet{GT21}, \citet{FST21}, \citet{GhodsiHSSY22}, and \citet{seddighin2022improved}.}
    \label{tab:general-valuations:MMS}
\end{table}

Regarding EFX allocations, \citet{PR18} proved that $1/2$-EFX allocations always exist for general valuations, and \citet{ChanCLW19} showed that computing such allocations can be done in polynomial time. Recently, \citet{journals/corr/abs-2205-07638} proved that, for instances with three agents, if one of them has an additive valuation and the others have general valuations, an exact EFX allocation always exists. 

\subsection{Arbitrary Entitlements} 
So far, we assumed that all agents have equal entitlements over the goods. However, there are settings where the fairness of an allocation must be considered with respect to asymmetric entitlements; for example, in many inheritance scenarios, closer relatives have higher entitlements than distant ones. 
Formally, in the \emph{weighted} or \emph{asymmetric} setting, each agent $i$ has an entitlement $0 < s_i < 1$ such that $\sum_{i\in N} s_i = 1$.
The traditional setting with symmetric agents is a special case when $s_i = 1/n$ for every $i$.
To capture fairness in the presence of arbitrary entitlements, one can generalize existing notions to their weighted counterparts, like \emph{weighted} EF1 and \emph{weighted} MMS. 

\begin{defn} [Weighted EF and EF1]
An allocation $A$ is {\em weighted envy-free} (WEF) if ${v_i(A_i)}/{s_i} \geq {v_i(A_j)}/{s_j}$ for every pair of agents $i, j \in N$. 
It {\em weighted envy-free up to one item} (WEF1) if ${v_i(A_i)}/{s_i} \geq {v_i(A_j \setminus\{g\})}/{s_j}$ for every pair of agents $i, j \in N$ and some $g\in A_j$.
\end{defn}

\noindent
\newtext{The concept of weighted fairness is based on ideas that were previously presented in, for example, \citep{Steinhaus49,books/daglib/0017730}.}
\citet{chakraborty2021weighted} proved that WEF1 allocations always exist and can be computed by a weight-dependent sequential picking algorithm that generalizes the Round-Robin algorithm.

Adapting MMS to the weighted setting is not as straightforward as WEF1.
The underlying idea is to define a {\em fairness ratio} for each agent.
Ideally, when the items are divisible, the fairest partition for agent $i$ is such that $v_i(A_j) = s_j \cdot v_i(M)$ for all $j$ (i.e., it is weighted proportional).
However, since the items are not divisible, it is inevitable to induce some degree of unfairness and the best fairness ratio for agent $i$ is
$$
F_i = \max_{A \in \mathcal{A}_n(M)} \min_{j \in N} \left\{ \frac{v_i(A_j)}{s_j \cdot v_i(M)} \right\}.
$$
Accordingly, the weighted MMS can be defined as follows.

\begin{defn} [Weighted MMS~\citep{FarhadiGHLPSSY19}]
An allocation $A$ is {\em weighted MMS} (WMMS) if $v_i(A_i) \geq \mu_i^n = F_i \cdot s_i \cdot {v_i(M)}$ for every agent $i \in N$. For any $\alpha \in [0,1]$, an allocation $A$ is $\alpha$-WMMS if $v_i(A_i) \geq \alpha\cdot \mu_i^n$ for every agent $i \in N$.
\end{defn}

\noindent
It can be verified that $F_i \le 1$ and $\mu_i^n \le s_i \cdot {v_i(M)}$.
Moreover, when the agents are symmetric, WMMS coincides with MMS.
\citet{FarhadiGHLPSSY19} proved that with arbitrary entitlements, the best approximation ratio of WMMS is $1/n$, which is somewhat surprisingly guaranteed by the Round-Robin algorithm. 

Novel fairness notions which highlight different perspectives of the weighted setting, such as {\em $\ell$-out-of-$d$ share} and {\em AnyPrice} share (APS), were proposed and studied by \citet{BabaioffNT21,BabaioffEF21}. 
\newtext{In particular, APS is defined as the maximum value an agent can guarantee to herself if she has a budget equal to her entitlement and the goods are adversarially priced (with prices that sum up to 1). In the unweighted setting, the APS value is at least as much as that of MMS, and sometimes it is strictly larger.}
\citet{BabaioffEF21} showed how to efficiently compute an allocation where every agent gets value no less than $3/5$ of her APS.
\citet{aziz2020propx} considered weighted PROP1 and presented an algorithm for computing a weighted PROP1 and PO allocation for indivisible items.

\subsection{Group Fairness}
In the model we discussed in the main part of the survey, we care about individual fairness.
However, there are practical applications where it makes more sense for the allocation to be group-wisely fair.
Depending on whether the groups of agents pre-exist or not, there are two lines of research on group fairness.
When the groups of agents do not pre-exist, we aim to compute allocations that are fair for any group of agents. 
For divisible items, \citet{berliant1992fair} defined the notion of {\em group envy-freeness} (GEF) by restricting envy-freeness to hold for groups of agents with equal size; clearly, this is strictly stronger than the requirement for EF to hold for pairs of agents. \citet{conitzer2019groups} adapted and extended GEF for indivisible items, leading to the notion of group-fairness (GF), which takes into account groups of different size. 

\begin{defn}[GF]
An allocation $A$ is {\em group-fair} (GF) if, for any non-empty sets of agents $S, T \subseteq N$ and every partition
$(B_i)_{i\in S}$ of $\bigcup_{j\in T} A_j$, $( (|S|/|T|)\cdot v_i(B_i))_{i\in S}$ does not Pareto
dominate $(v_i(A_i))_{i\in S}$.
\end{defn}

From the above definition, we can see that the two extreme cases where $|S| = |T| = 1$ and $|S|=|T| =n$ imply the traditional requirements of EF and PO. \citet{conitzer2019groups} and \citet{conf/ijcai/0001R20} introduced different variants of ``up to one'' relaxations of GEF/GF and proved existential and hardness results.

When the agents are previously partitioned into several groups (e.g., each group might correspond to a family), then we only need to satisfy each given group of agents. Several models capturing scenarios along these lines have been considered in the literature. \citet{suksompong2018groups} focused on a setting where each agent derives full value from all the goods allocated to the group she belongs to, and showed bounds on the best possible approximation of MMS. \citet{kyropoulou2020groups} considered EF1 and EFX allocations in the same setting, as well as in settings with dynamic group formation; some of these results were later improved by \citet{manurangsi2021groups} using ideas from discrepancy theory. \citet{halevi2019democratic} focused on the case of democratic fairness, where the goal is to compute allocations that are considered fair (e.g., satisfying EF1) by a large fraction of the agents in each group.

\subsection{Randomness in Fair Division}
\label{sec:randomness}
\paragraph{Best of both worlds.} 
We mentioned in the introduction that randomization can be used as a tool to circumvent the impossibilities in the task of allocating indivisible resources. The guarantee that one can obtain this way however is that of \emph{ex-ante} EF (EF with respect to the expected values). In fact, in the absence of any other requirement, this is rather trivial; simply allocate the goods to the agents uniformly at random. 
On the other hand, this sort of approach cannot help with achieving \emph{ex-post} fairness, i.e., fairness in each resulting allocation, for any of the deterministic notions of fairness mentioned in this survey. 
Therefore, it seems natural to ask for a randomized allocation algorithm that also has good ex-post fairness guarantees.

As we have seen, Round-Robin algorithms return an EF1 allocation, but the agents have an advantage over those who are ordered after them. To eliminate this unfair advantage, we can consider a randomized version where the ordering of the agents is picked at random. \citet{freeman2020rand} showed that this randomized Round-Robin algorithm guarantees ex-ante PROP but not ex-ante EF.

The Probabilistic-Serial (PS) algorithm of~\citet{journals/jet/BogomolnaiaM01} is a randomized algorithm for allocating indivisible items in an ex-ante EF manner. Agents eat their most preferred good at a uniform rate and, once consumed, move on to the next available good. 
The algorithm terminates when all goods have been consumed, and the probability share of an agent for a good is the fraction of the good eaten by the agent. 
\citet{freeman2020rand} presented a variant of this procedure, the recursive probabilistic serial algorithm, that maintains the ex-ante EF property, and moreover every realized allocation is ex-post EF1.
Later, \citet{aziz2020rand} showed that the random allocation generated by the original probabilistic serial algorithm is equivalent to lottery over a set of EF1 allocations.
\citet{journals/corr/abs-2008-08991} presented an eating algorithm that is suitable for any type of feasibility constraint and allocation problem with ordinal preferences. 

Finally, in a somewhat different direction, \citet{caragiannis2021interim} studied \emph{interim} envy-freeness, a notion which lies between ex-ante and ex-post envy-freeness. They showed positive and impossibility results for many settings, including one-sided matching, where each agent can be given only one good.  

\paragraph{Stochastic Settings.}
Although it is easy to construct instances where an envy-free allocation does not exist, there is a line of research showing that when the values are randomly drawn from some probability distributions (instead of being chosen adversarially), an envy-free allocation is likely to exist as long as the number of items is sufficiently large compared to the number of agents.

In particular, \citet{conf/aaai/DickersonGKPS14} showed that when the values of the agents are independently drawn from an identical distribution, the social welfare maximizing allocation (that is, allocating each good to the agent who has the highest value for it) is envy-free with high probability when $m = \Omega(n\log n)$.
\citet{journals/siamdm/ManurangsiS20,journals/siamdm/ManurangsiS21} further improved the lower bound on the existence of envy-free allocations, and showed that the Round-Robin algorithm produces an envy-free
allocation for a slightly lower $m$.
Recently, \citet{bai2022envy} extended this result to the case of asymmetric distributions.  

In a different direction, \citet{bai2022fair} showed that given a worst-case instance that does not admit any envy-free allocation, randomly perturbing the values of each agent leads to the existence of envy-free allocations with high probability.

\subsection{Online Fair Division}
Most models we have discussed are static, as all goods, agents, and their valuation functions do not change over time. Online fair division considers settings where the agents or the goods arrive in an online manner.
We briefly mention some of the most related works; we refer the readers to \citet{AleksandrovAGW15,aleksandrov2020onlinesurvey} for a more detailed discussion.  
In the most common model, $T$ items arrive online; when an item arrives, the values of the agents for it are realized and based on these we need to decide immediately and irrevocably how to allocate it, typically without knowing future events (including how many items will arrive later and the values of the agents for these items).

Let $\textrm{Envy}_T$ be the cumulative envy until time $T$, i.e., the maximum difference between any agent’s value for goods allocated to another agent and to herself in the first $T$ rounds. 
Then, it is desired that the envy can vanish as time goes on, that is, $\textrm{Envy}_T/T \to 0$ when $T$ goes to infinity. \citet{BenadeKPP18} proved that by allocating the new good to an agent chosen uniformly at random, $\textrm{Envy}_T = \Tilde{O}(\sqrt{T/n})$, and thus we can indeed eliminate the envy over time.
\citet{BenadeKPP18} also showed that there exists a deterministic polynomial-time algorithm with the same envy bound as the random allocation algorithm, and that this bound is asymptotically optimal against an adaptive adversary (who can see the random bits of the algorithm through the first $T$ goods and decide the value $v_{iT}$), meaning that the allocation is far from being EF1. To circumvent this obstacle, \citet{HePPZ19} relaxed the requirement of ``irrevocable decision'' by allowing the algorithm to reallocate previously allocated items. Clearly, if we reallocate all goods at each time (i.e., with $\Theta(T^2)$ reallocations), we can surely achieve EF1, but this makes the online setting meaningless. 
\citet{HePPZ19} showed that with two agents we actually only need $\Theta(T)$ reallocations; for more than two agents, $O(T^{3/2})$ reallocations are sufficient and $\Omega(T)$ are necessary. 

The stochastic setting of online fair allocation is also well-studied.  
Actually, the algorithms by \citet{conf/aaai/DickersonGKPS14}, \citet{KPW16} and \citet{bai2022envy} that we have introduced in Section \ref{sec:randomness} for allocating goods with random values do not need to access the order of goods and also work for the online setting. 
Thus, we can achieve envy-freeness with high probability by allocating each arriving good to the agent with the highest value. 
\citet{zeng2020dynamic} further showed that, with some modification, the returned allocation is either EF1, or envy-free with high probability. 
\citet{zeng2020dynamic} also investigated the trade-off between fairness and efficiency when the adversary has different levels of power.

It remains an open question whether there exist competitive online algorithms for the computation of (approximately) MMS or Prop1/PropX allocations. For the case of identical valuation functions, approximate MMS allocations correspond to maximizing the minimum load on the job scheduling problem, for which optimal competitive ratios have been proved by~\citet{AzarE97} and~\citet{TanW07}.
The alternative model that considers a fixed set of resources and agents who arrive or depart over time has not been considered for indivisible resources, partially because it is very challenging to achieve positive results \citep{KashPS14}.

\subsection{Subsidies}
As we saw in Section \ref{subsec:strategic}, even in a game-theoretic setting no monetary transfers are allowed in fair division problems. Indeed, arbitrary payments would significantly alter the flavor of these problems and often go against their motivation. A recent line of work, however, considers the question of whether it is possible to pay the agents just a small amount of money (subsidy) on top of a given allocation in order to make it envy-free (when the subsidies are also taken into consideration). Allocations for which this can be done are called \emph{envy-freeable}.
Formally, a pair $(A,p)$ consisting of an allocation $A=(A_1,\dots,A_n)$ and a payment profile $p=(p_1,\dots,p_n)$, where $p_i \ge 0$ for every $i$, is envy-free if $v_i(A_i)+p_i \ge v_i(A_j)+p_j$. An allocation $A$ is {\em envy-freeable} if there are payments $p=(p_1,\dots,p_n)$ such that $(A,p)$ is envy-free.

\citet{halpern2019subsidies} first noted that not all allocations are envy-freeable by considering an example with two agents that have values $150$ and $100$ for a single good. 
If the good is allocated to the second agent, then we have to compensate the first agent at least $150$, however, this would make the second agent envious.
\citet{halpern2019subsidies} then gave a characterization of all envy-freeable allocations. Specifically, they showed that allocation $A$ is envy-freeable if and only if $A$ maximizes the utilitarian welfare across all reassignments of its bundles to agents. 

A natural question to ask in this setting is to quantify the minimum total amount of subsidy $\sum_i p_i$ required to find an envy-free allocation. Towards answering this question, \citet{halpern2019subsidies} showed that, if a fixed envy-freeable allocation is given, the minimum total subsidy can be computed in polynomial time. However, finding the envy-freeable allocation with overall minimum total subsidy is an NP-hard problem; this follows since deciding the existence of an envy-free allocation without any subsidy is, as we already discussed, NP-hard~\citep{journals/jair/BouveretL08}. \citet{caragiannis2021subsidies} provided (additive) approximation guarantees and hardness results for computing an envy-freeable allocation that minimizes the total amount of subsidies. 

To bound the minimum subsidy required in the worst case over allocations, \citet{halpern2019subsidies} proved that any envy-freeable allocation requires no more than $(n - 1)m\cdot v^*$ total subsidy, where $v^*$ is the maximum value of an agent for a good. If we are able to choose the allocations, the $(n - 1)$ factor cannot be removed.
Consider an instance with a single good that all agents value as $v^*$; then, a total payment of at least $(n-1)\cdot v^*$ is required to eliminate the envy of all agents, in any allocation. \citet{halpern2019subsidies} conjectured that it is always possible to find an allocation with no more than $(n-1)\cdot v^*$ total subsidy, which was later indeed proved by \citet{brustle2020dollar}. In addition, \citeauthor{brustle2020dollar} proved that it suffices to subsidize each agent by at most $v^*$.

\subsection{Mixtures of Indivisible and Divisible Items}
There has also been recent work on models that involve both indivisible and divisible goods; note that limited subsidy can be thought of as a divisible good. \citet{DBLP:journals/ai/BeiLLLL21} proposed a new fairness notion called \textit{envy-freeness for mixed goods (EFM)}, which is a generalization of both EF and EF1 to the mixed goods setting. 
\newtext{The key idea behind EFM is that (1) an agent $i$'s envy towards another agent $j$ vanishes, if an indivisible good is removed from consideration when $j$ only gets indivisible goods, or (2) agent $i$ does not envy agent $j$ when $j$ gets some divisible goods.}
\citet{DBLP:journals/ai/BeiLLLL21} proved that an EFM allocation always exists for any number of agents with additive valuations. \citet{DBLP:journals/aamas/BeiLLW21} examined the same setting and explored approximations of MMS.
\citet{BhaskarSV21} showed that for a mixed resources model consisting of indivisible items and a divisible, undesirable heterogeneous resource, an EFM allocation always exists.

\subsection{Chores and Mixed Manna} 
Beyond discrete fair division of goods, there is a significant line of work that considers similar questions when items can be seen as \emph{chores} (which are negatively valued by the agents), or \emph{mixed manna} (a mixture of both goods and chores). Here, we give a brief overview of these settings.

The definitions of EF, PROP and MMS remain the same as in Definitions~\ref{def:ef}, \ref{def:prop}, \ref{def:mms}, and \ref{def:mms:app}. As in the case of goods, the existence of EF or PROP allocations is rarely guaranteed. For additive valuations, \citet{AzizRSW17} proved that the Round-Robin algorithm returns a $2$-MMS allocation. This result was later improved to $4/3$ and $11/9$ by \citet{BarmanK20} and \citet{HuangL21}, respectively. 
On the negative side, \citet{FST21} proved that no algorithm can ensure an approximation factor better than $44/43$. 
If the algorithm has access only to the ordinal preferences of the agents, the best approximation of MMS that can be guaranteed is between $7/5$ and $5/3$~\citep{aziz2022approximate}. When the valuations are submodular, \citet{journals/corr/abs-2205-10520} proved that no algorithm guarantees a better than $n$-MMS allocation.
\citet{conf/ijcai/0001C019} considered the case when the agents have asymmetric entitlements and extended the notion of weighted MMS to chores, but the tight approximation ratio is still unknown. 

Relaxations of EF and PROP need to be changed a bit to capture the fact that the agents have negative values for chores. Essentially, instead of removing goods from bundles of other agents to eliminate the envy of an agent, we need to remove chores from the bundle of the agent herself. 

\begin{defn}[EF1 and EFX for chores]
For chores, an allocation $A$ is $\alpha$-EF1 if $v_i(A_i \setminus \{g\}) \ge \alpha \cdot v_i(A_j)$ for any pair of agents $i,j$ and some $g \in A_i$; $A$ is $\alpha$-EFX if the inequality holds for any $g \in A_i$.
\end{defn}

\noindent 
For additive valuations, an EF1 allocation can be easily computed by the Round-Robin algorithm.
However, as noted by \citet{BhaskarSV21}, the allocation returned by the Envy-Cycle Elimination algorithm may not be EF1 if the cycles are resolved arbitrarily. Instead, the top-trading technique (in which each agent only points to the agent she envies the most in the envy graph) must be used to preserve EF1. 
In contrast to the case of goods where EF1 and PO can be satisfied simultaneously, the compatibility of EF1 and PO is still unknown for chores.
\citet{conf/atal/EbadianP022}, \citet{conf/aaai/GargMQ22} and \citet{journals/corr/abs-2205-11363} proved that, for bi-valued instances, EF1+PO allocations always exist and can be found efficiently; this is the only known result so far.
The existence of EFX allocations for chores also remains open, with the only positive result being the computation of $O(n^2)$-EFX allocations due to \citet{ijcai/ZhouW22}.

\begin{defn}
For chores, an allocation $A$ is Prop1 if $v_i(A_i \setminus \{g\}) \ge \alpha \cdot  v_i(M)/n$ for any pair of agents $i,j$ and some $g \in A_i$; $A$ is PropX if the inequality holds for any $g \in A_i$.
\end{defn}

\noindent 
The existence of Prop1 allocations is straightforward as EF1 implies Prop1.
But the good news for Prop1 is that it can always be satisfied together with PO, even if the set of items is a mixture of goods and chores~\citep{aziz2020propx}.
In contrast to goods, where PropX allocations may not exist, for chores, a PropX allocations always exist and can be computed efficiently, even when the agents have asymmetric entitlements~\citep{BoYX}.
However, it is still unknown whether PropX and PO are compatible or not. The relationships among various fairness notions for chores are discussed by \citet{SunCD21}.

The more general case of mixtures of goods and chores has recently also been studied~ \citep{BogomolnaiaMSY17,journals/scw/BogomolnaiaMSY19a,aziz2020propx,AzizCIW22}.
This model is particularly interesting as it includes the setting with non-monotone valuations. 
\citet{AzizCIW22} proved that a double Round-Robin algorithm computes an EF1 allocation for any number of agents, and a generalized adjusted winner algorithm finds an EF1+PO allocation for two agents. 
Recently, \citet{aziz2020propx} and \citet{conf/sigecom/KulkarniMT21} designed algorithms for computing Prop1+PO or approximately MMS+PO allocations, respectively.
It is an intriguing future research direction to study the fair allocation problem under other non-monotonic valuations. 
 
Finally, the fair allocation of divisible chores is also studied in the literature, such as the computation of envy-free allocations \citep{conf/soda/DehghaniFHY18} and competitive equilibria \citep{DBLP:conf/soda/BoodaghiansCM22,conf/soda/ChaudhuryGMM21}.
As such settings are out of the scope of our survey, we refer the reader to these works and references therein.

\subsection{Other Settings}
There are several other variants that we do not discuss in the current survey. Depending on the application, some allocations may not be feasible due to various restrictions, such as connectivity, cardinality, separation, or budget constraints. Such models have recently attracted the attention of the community. Rather than referring to specific works, we point the reader to the survey of \cite{suksompong2021constraints} which discusses this part of the literature in detail. For another example, externalities have been studied in fair division by \citet{journals/scw/SeddighinSG21} and \citet{journals/corr/abs-2110-09066}, where the agents not only have value for what they get but also for what others get.
Yet another example is related to public resource allocation where the items can be shared by the agents ~\citep{conitzer2017prop1,conf/sigecom/FainM018,conf/fsttcs/GargKM21}.  

\section*{Acknowledgements}
This work is partially supported by 
the ERC Advanced Grant 788893 AMDROMA ``Algorithmic and Mechanism Design Research in Online Markets'', 
the MIUR PRIN project ALGADIMAR ``Algorithms, Games, Digital Markets'', 
the NWO Veni project No.~VI.Veni.192.153,
NSFC No. 62102333, 
HKSAR RGC No. PolyU 25211321, 
PolyU Start-up No. P0034420, 
FDCT (File no. 0014/2022/AFJ, 0085/2022/A, 0143/2020/A3, SKL-IOTSC-2021-2023).

\bibliographystyle{plainnat}
\bibliography{references}

\end{document}